\documentclass[runningheads, envcountsame, a4paper]{llncs}

\usepackage[hidelinks]{hyperref}
\usepackage{amssymb} %
\usepackage{mathtools}
\usepackage[vlined, noend, linesnumbered]{algorithm2e}
\usepackage{xspace}
\usepackage{subcaption}
\usepackage{tikz}
\usetikzlibrary{backgrounds, automata, positioning, arrows, plotmarks}
\usepackage{pgfplots}
\pgfplotsset{compat=1.17}

\makeatletter
\renewcommand\paragraph{\@startsection{paragraph}{4}{\z@}%
                     {-12\p@ \@plus -4\p@ \@minus -4\p@}%
                     {-0.5em \@plus -0.22em \@minus -0.1em}%
                     {\normalfont\normalsize\bfseries}}
\makeatother
\newcommand{\tmprel}{\ltimes}
\newcommand{\defeq}{\triangleq}%
\newcommand{\defrel}{\stackrel{{\mbox{\tiny\ensuremath{\triangle}}}}{\Longleftrightarrow}} %
\newcommand{\st}{\mid} %

\newcommand{\run}[2][\ensuremath{-\!\!\!\longrightarrow}]{\mathrel{\substack{\raisebox{-3.8pt}[0pt][0pt]{\scalebox{.88}{\ensuremath{#2}}}\\\raisebox{-3.8pt}[0pt][0pt]{\ensuremath{\substack{#1}}}}}}

\tikzstyle{state}=[thick,minimum size=18pt, circle,draw]
\tikzstyle{transition}=[->,thick,>=stealth,shorten >=1pt,shorten <=1pt]
\tikzstyle{final}=[after node path={ node[state, scale=.8] at (\tikzlastnode) {} }]
\tikzstyle{initial}=[after node path={
	[to path={[transition] (\tikztostart) -- (\tikztotarget)}]
	(\tikzlastnode)++(180:22pt) edge (\tikzlastnode)
}]
\tikzset{
	bg/.default={},
	bg/.style={execute at end picture={
			\begin{scope}[on background layer]
				\node[xshift=-1mm, yshift=-1mm] (sw) at (current bounding box.south west) {};
				\node[xshift=1mm, yshift=1mm] (ne) at (current bounding box.north east) {};
				\node[xshift=1mm, yshift=-1mm] (nw) at (current bounding box.north west) {};
				\fill[fill=black!10,rounded corners] (sw) rectangle (ne);

				\ifx&#1&\else
				\node[anchor=north east, xshift=2pt] at (nw) {#1};
				\fi
			\end{scope}
	}},
}

\newcommand{\at}[1]{\scalebox{1.5}{.}{#1}} %
\newcommand{\super}[1]{\wp(#1)}
\newcommand{\lasso}[2]{#1^{}#2^{\omega}}

\newcommand{\tuple}[1]{\langle{#1}\rangle}

\newcommand{\ok}{\ensuremath{\mathtt{ok}}\xspace}
\newcommand{\ko}{\ensuremath{\mathtt{ko}}\xspace}
\newcommand{\textRcat}{\ensuremath{\mathtt{Rcat}}\xspace}
\newcommand{\textRcatA}{\ensuremath{\textRcat_\calA}\xspace}
\newcommand{\textRunI}{\ensuremath{\mathtt{Tgt}}\xspace}
\newcommand{\textRunIB}{\ensuremath{\textRunI_\calB}\xspace}
\newcommand{\textRunF}{\ensuremath{\mathtt{Cxt}}\xspace}
\newcommand{\textRunFB}{\ensuremath{\textRunF_\calB}\xspace}

\newcommand{\calA}{\ensuremath{\mathcal{A}}\xspace}
\newcommand{\calB}{\ensuremath{\mathcal{B}}\xspace}
\newcommand{\calF}{\ensuremath{\mathcal{F}}\xspace}

\newcommand{\base}[3][]{\ensuremath{\mathfrak{B}_{#1}(#2, #3)}\xspace}
\newcommand{\stem}{\ensuremath{\mathtt{Stem}}\xspace}
\newcommand{\per}{\ensuremath{\mathtt{Per}}\xspace}
\newcommand{\UltimA}{\ensuremath{\texttt{Ultim}_\calA}\xspace}

\newcommand{\zero}{\bot}
\newcommand{\one}{\top}
\newcommand{\N}{\mathbb{N}}
\newcommand{\BA}{\textsf{BA}\xspace}
\newcommand{\BAs}{\textsf{BAs}\xspace}
\newcommand{\FORQ}{\textsf{FORQ}\xspace}
\newcommand{\FORQs}{\textsf{FORQs}\xspace}
\newcommand{\FORKLIFT}{\textsc{Forklift}\xspace}

\newcommand{\catA}[2][]{\ensuremath{\textRcatA^{#1}(#2)}}
\newcommand{\up}[2]{\ensuremath{{}_{#1}\!\scalebox{1.2}{\({\upharpoonleft}\)}{#2}}}
\newcommand{\runI}[1]{\ensuremath{\textRunI(#1)}}
\newcommand{\runIB}[1]{\ensuremath{\textRunI_\calB(#1)}}
\newcommand{\runF}[2][]{\ensuremath{\textRunF(#1, #2)}}
\newcommand{\runFB}[2][]{\ensuremath{\textRunF_\calB(#1, #2)}}

\newcommand{\leqI}{%
	\mathrel{%
		\mathchoice%
		{\ensuremath{\raisebox{-3.5pt}[.1pt][.1pt]{\(\stackrel{\mbox{\scalebox{.95}{\(<\)}}}{\mbox{\scalebox{.85}{\(\sim\)}}}\)}}} %
		{\ensuremath{\raisebox{-3.5pt}[.1pt][.1pt]{\(\stackrel{\mbox{\scalebox{.95}{\(<\)}}}{\mbox{\scalebox{.85}{\(\sim\)}}}\)}}} %
		{\ensuremath{\raisebox{-2.5pt}[.1pt][.1pt]{\(\stackrel{\mbox{\scalebox{.75}{\(<\)}}}{\mbox{\scalebox{.65}{\(\sim\)}}}\)}}} %
		{\ensuremath{\raisebox{-2pt}[.1pt][.1pt]{\(\stackrel{\mbox{\scalebox{.55}{\(<\)}}}{\mbox{\scalebox{.45}{\(\sim\)}}}\)}}} %
	}
}

\newcommand{\leqF}{%
	\mathrel{%
		\mathchoice%
		{\ensuremath{\raisebox{-3.5pt}[.1pt][.1pt]{\(\stackrel{\mbox{\scalebox{.85}{\(\ll\)}}}{\mbox{\scalebox{.95}{\(\sim\)}}}\)}}} %
		{\ensuremath{\raisebox{-3.5pt}[.1pt][.1pt]{\(\stackrel{\mbox{\scalebox{.85}{\(\ll\)}}}{\mbox{\scalebox{.95}{\(\sim\)}}}\)}}} %
		{\ensuremath{\raisebox{-2.5pt}[.1pt][.1pt]{\(\stackrel{\mbox{\scalebox{.65}{\(\ll\)}}}{\mbox{\scalebox{.75}{\(\sim\)}}}\)}}} %
		{\ensuremath{\raisebox{-2pt}[.1pt][.1pt]{\(\stackrel{\mbox{\scalebox{.45}{\(\ll\)}}}{\mbox{\scalebox{.55}{\(\sim\)}}}\)}}} %
	}
}

\usepackage{todonotes}

\begin{document}
\title{FORQ-based Language Inclusion Formal Testing\thanks{%
		This work was partially funded by the ESF Investing in your future, the Madrid regional project S2018/TCS-4339 BLOQUES, the Spanish project PGC2018-102210-B-I00 BOSCO, the Ramón y Cajal fellowship RYC-2016-20281, and the ERC grant PR1001ERC02.%
}}

\author{
	Kyveli~Doveri\inst{1,2}\orcidID{0000-0001-9403-2860} \and
	Pierre~Ganty\inst{1}\orcidID{0000-0002-3625-6003} \and
	Nicolas~Mazzocchi\inst{1,3}\orcidID{0000-0001-6425-5369}
}
\institute{
	IMDEA Software Institute, Madrid, Spain\\\email{firsname.lastname@imdea.org} \and Universidad Politécnica de Madrid, Spain \and
	Institute of Science and Technology Austria, Klosterneuburg, Austria
}
	\authorrunning{K. Doveri, P. Ganty, and N. Mazzocchi}

	\maketitle

	\begin{abstract}%
		We propose a novel algorithm to decide the language inclusion between (nondeterministic) B\"uchi automata, a \textsc{PSpace}-complete problem.
    Our approach, like others before, leverage a notion of quasiorder to prune the search for a counterexample by discarding candidates which are subsumed by others for the quasiorder.
    Discarded candidates are guaranteed to not compromise the completeness of the algorithm.
    The novelty of our work lies in the quasiorder used to discard candidates.
    We introduce FORQs (family of right quasiorders) that we obtain by adapting the notion of family of right congruences put forward by Maler and Staiger in 1993.
    We define a FORQ-based inclusion algorithm which we prove correct and instantiate it for a specific FORQ, called the structural FORQ, induced by the B\"uchi automaton to the right of the inclusion sign.
    The resulting implementation, called \textsc{Forklift}, scales up better than the state-of-the-art on a variety of benchmarks including benchmarks from program verification and theorem proving for word combinatorics.
		\keywords{Language inclusion \and  B\"uchi automata \and Well-quasiorders}
    \textbf{Artifact:} \url{https://doi.org/10.5281/zenodo.6552870}
	\end{abstract}

\section{Introduction}\label{sec:introduction}
In verification \cite{DBLP:conf/tacas/HeizmannCDGHLNM18,DBLP:conf/cav/HeizmannHP13} and theorem proving \cite{DBLP:journals/corr/abs-2102-01727}, B\"uchi automata have been used as the underlying formal model.
In these settings, B\"uchi automata respectively encode 1) the behaviors of a system as well as properties about it; and 2) the set of valuations satisfying a predicate.
Questions like asking whether a system complies with a specification naturally reduce to a language inclusion problem and so does proving a theorem of the form \(\forall x\, \exists y,\, P(x) \Rightarrow Q(y)\).

In this paper we propose a new algorithm for the inclusion problem between \(\omega\)-regular languages given by B\"uchi automata.
The problem is \textsc{PSpace}-complete~\cite{kupfermanVerificationFairTransition1996a} and significant effort has been devoted to the discovery of algorithms for inclusion that behave well in practice \cite{DBLP:conf/concur/AbdullaCCHHMV11,DBLP:journals/lmcs/ClementeM19,Doyen2009,fogartyEfficientBuchiUniversality2010a,DBLP:journals/fuin/KuperbergPP21,DBLP:conf/tacas/LiSTCX19}.
Each proposed algorithm is characterized by a set of techniques (e.g. Ramsey-based, rank-based) and heuristics (e.g. antichains, simulation relations).
The algorithm we propose falls into the category of Ramsey-based algorithms and uses the antichain \cite{DBLP:conf/cav/WulfDHR06} heuristics: the search for counterexamples is pruned using quasiorders.
Intuitively when two candidate counterexamples are comparable with respect to some considered quasiorder, the ``higher'' of the two can be discarded without compromising completeness of the search.
In our setting, counterexamples to inclusion are ultimately periodic words, i.e., words of the form \(\lasso{u}{v}\), where \(u\) and \(v\) are called a \emph{stem} and a \emph{period}, respectively.
Therefore pruning is done by comparing stems and periods of candidate counterexamples during the search.

In the work of Abdulla et al.~\cite{DBLP:conf/cav/AbdullaCCHHMV10,DBLP:conf/concur/AbdullaCCHHMV11} which was further refined by Clemente et al.~\cite{DBLP:journals/lmcs/ClementeM19} they use a single quasiorder to compare both stems and periods.
Their effort has been focused on refining that single quasiorder by enhancing it with simulation relations.
Others including some authors of this paper, followed an orthogonal line~\cite{DBLP:conf/concur/DoveriGPR21,DBLP:journals/fuin/KuperbergPP21} that investigates the use of two quasiorders: one for the stems and another one, independent, for the periods.
The flexibility of using different quasiorders yields more pruning when searching for a counterexample.
In this paper, we push the envelope further by using an unbounded number of quasiorders: one for the stems and a family of quasiorders for the periods each of them depending on a distinct stem.
We use the acronym \FORQ, which stands for \emph{family of right quasiorders}, to refer to these quasiorders.
Using \FORQs leads to significant algorithmic differences compared to the two quasiorders approaches.
More precisely, the algorithms with two quasiorders \cite{DBLP:conf/concur/DoveriGPR21,DBLP:journals/fuin/KuperbergPP21} compute exactly two fixpoints (one for the stems and one for the periods) independently whereas the \FORQ-based algorithm that we present computes two fixpoints for the stems and unboundedly many fixpoints for the periods (depending on the number of stems that belong to the first two fixpoints).
Even though we lose the stem/period independence and we compute more fixpoints, in practice, the use of \FORQs scales up better than the approaches based on one or two quasiorders.

We formalize the notion of \FORQ by relaxing and generalizing the notion of \emph{family of right congruences} introduced by Maler and Staiger \cite{malerSyntacticCongruencesLanguages2008a} to advance the theory of recognizability of \(\omega\)-regular languages and, in particular, questions related to minimal-state automata.
Recently, families of right congruences have been used in other contexts like the learning of \(\omega\)-regular languages (see \cite{lmcs:4283} and references therein) and B\"uchi automata complementation \cite{liCongruenceRelationsBuchi2021}. %

Below, we describe how our contributions are organized:
\begin{itemize}
  \item We define the notion of \FORQs and leverage them to identify key finite sets of stems and periods that are sound and complete to decide the inclusion problem (Section~\ref{sec:crux}).
  \item We introduce a \FORQ called the structural \FORQ which relies on the structure of a given B\"uchi automaton (Section~\ref{sec:state}).
  \item We formulate a \FORQ-based inclusion algorithm that computes such key sets as fixpoints, and then use these key stems and periods to search for a counterexample to inclusion via membership queries (Section~\ref{sec:forq}).
  \item We study the algorithmic complexity of the \FORQ-based inclusion algorithm instantiated with structural \FORQs (Section~\ref{sec:complexity}).
  \item We implement the inclusion algorithm with structural \FORQs in a prototype called \FORKLIFT and we conduct an empirical evaluation on a set of 674 benchmarks (Section~\ref{sec:implementation}).
\end{itemize}

\section{Preliminaries}\label{sec:preliminaries}

\paragraph{Languages.}
Let \(\Sigma\) be a finite and non-empty \emph{alphabet}.
We write \(\Sigma^*\) to refer to the set of finite words over \(\Sigma\) and we write \(\varepsilon\) to denote the empty word.
Given \(u \in \Sigma^*\), we denote by \(|u|\) the length of \(u\).
In particular \(|\varepsilon| = 0\).
We also define \(\Sigma^+ \defeq \Sigma^* \setminus \{ \varepsilon\}\), and \(\Sigma^{\nabla n} \defeq \{ u\in\Sigma^* \st |u|\mathrel{\nabla} n\}\) with \(\nabla\in\{\leq,\geq\}\), hence \(\Sigma^* = \Sigma^{\geq 0}, \Sigma^+=\Sigma^{\geq 1}\).
We write \(\Sigma^{\omega}\) to refer to the set of infinite words over \(\Sigma\).
An infinite word \(\mu \in \Sigma^{\omega}\) is said to be \emph{ultimately periodic} if it admits a decomposition \(\mu=\lasso{u}{v}\) with \(u\in \Sigma^*\) (called a \emph{stem}) and \(v\in \Sigma^+\) (called a \emph{period}). We fix an alphabet \(\Sigma\) throughout the paper.

\paragraph{Order Theory.}
Let \(E\) be a set of elements and \(\mathord{\tmprel}\) be a binary relation over \(E\).
The relation \(\mathord{\tmprel}\) is said to be a \emph{quasiorder} when it is \emph{reflexive} and \emph{transitive}.
Given a subset \(X\) of \(E\), we define its \emph{upward closure} with respect to the quasiorder \(\mathord{\tmprel}\) by \(\up{\tmprel}{X} \defeq \{ e \in E \mid \exists x\in X, x\tmprel e\}\).
Given two subsets \(X, Y\subseteq E\) the set \(Y\) is said to be a \emph{basis} for \(X\) with respect to \(\mathord{\tmprel}\), denoted \({\base[\tmprel]{Y}{X}}\), whenever \(Y\subseteq X\) and \(\up{\tmprel}{X}=\up{\tmprel}{Y}\).
The quasiorder  \(\mathord{\tmprel}\) is a \emph{well-quasiorder} if{}f for each set \(X\subseteq E\) there exists a finite set \(Y\subseteq E\) such that \(\base[\tmprel]{Y}{X}\).
This property on bases is also known as the \emph{finite basis property}.
Other equivalent definitions of well-quasiorders can be found in the literature~\cite{DBLP:journals/acta/LucaV94}, we will use the followings:
\begin{enumerate}
  \item For every \(\{e_i\}_{i\in\N}\in E^\N\) there exists \(i,j\in\N\) with \(i<j\) such that \(e_i \tmprel e_j\).
  \item No sequence \(\{X_i\}_{i\in\N}\in \super{E}^\N\) is such that \(\up{\mathord{\tmprel}}{X_1} \subsetneq \up{\mathord{\tmprel}}{X_2} \subsetneq \ldots\) holds.\footnote{The notation \(\wp(E)\) denotes the set of all subsets of \(E\).}
\end{enumerate}

\paragraph{Automata.}
A (nondeterministic) \emph{B\"uchi automaton} \calB (\BA for short) is a tuple $(Q, q_I, \Delta, F)$ where $Q$ is a finite set of states including \(q_I\), the initial state, $\Delta \subseteq Q \times \Sigma \times Q$ is the transition relation, and, $F \subseteq Q$ is the set of accepting states.
We lift $\Delta$ to finite words as expected. %
We prefer to write \(\calB \colon q_1 \run{u} q_2\) instead of $(q_1, u, q_2) \in \Delta$.
In addition, we write \(\calB \colon q_1 \run{u}_F q_2 \) when there exists a state \(q_F\in F\) and two words \(u_1, u_2\) such that \(\calB \colon q_1 \run{u_1} q_F \run{u_2} q_2\), and \(u=u_1u_2\).

A run $\pi$ of \calB over $\mu = a_0 a_1 \cdots  \in \Sigma^\omega$ is a function $\pi \colon \N \rightarrow Q$ such that $\pi(0) = q_I$ and for all position $i \in \N$, we have that $\calB\colon\pi(i) \run{a_i} \pi(i+1)$.
A run is said to be \emph{accepting} if $\pi(i) \in F$
for infinitely many values of \(i\in\N\).
The language  \(L(\calB)\) of words \emph{recognized} by \calB is the set of $\omega$-words for which \calB admits an accepting run.
A language \(L\) is \(\omega\)-\emph{regular} if it is recognized by some \BA.

\section{Foundations of our Approach}\label{sec:crux}

Let \(\calA \defeq (P, p_I, \Delta_{\calA}, F_{\calA})\) be a B\"uchi automaton and \(M\) be an \(\omega\)-regular language.
The main idea behind our approach is to compute a finite subset \(T_\calA\) of ultimately periodic words of \(L(\calA)\) such that:
\begin{equation}\tag{\(\dagger\)}\label{eq:crux:equivalence}
	T_\calA \subseteq M \iff L(\calA) \subseteq M\enspace.
\end{equation}
Then \(L(\calA) \subseteq M\) holds if{}f each of the finitely many words of \(T_\calA\) belongs to \(M\) which is tested via membership queries.

First we observe that such a subset always exists: if the inclusion holds take \(T_\calA\) to be any finite subset of \(L(\calA)\) (empty set included); else take \(T_\calA\) to contain some ultimately periodic word that is a counterexample to inclusion. 
In what follows, we will show that a finite subset \(T_\calA\) satisfying \eqref{eq:crux:equivalence} can be computed by using an ordering to prune the ultimately periodic words of \(L(\calA)\).
We will obtain such an ordering using a family of right quasiorders, a notion introduced below.
\begin{definition}[\FORQ]
A family of right quasiorders is a pair \(\tuple{\mathord{\leqI}, \{\mathord{\leqF_{u}}\}_{u \in \Sigma^{*}}}\) where \(\mathord{\leqI} \subseteq \Sigma^* \times \Sigma^*\) is a right-monotonic\footnote{A quasiorder \(\mathord{\tmprel}\) on \(\Sigma^*\) is \emph{right-monotonic} when \(u \tmprel v\) implies \( u\, w \tmprel v\, w\) for all \(w\in\Sigma^*\).}
  quasiorder as well as every \(\mathord{\leqF_{u}} \subseteq\Sigma^* \times \Sigma^*\) where \(u \in \Sigma^*\).
Additionally, for all \(u, u' \in  \Sigma^*\), we require \(u \leqI u' \Rightarrow \mathord{\leqF_{u'}} \subseteq \mathord{\leqF_u}\) called the \emph{\FORQ constraint}.
\end{definition}
First, we observe that the above definition uses separate orderings for stems and periods.  
The definition goes even further, the ordering used for periods is depending on stems so that a period may or may not be discarded depending on the stem under consideration.
The FORQ constraint tells us that if the periods \(v\) and \(w\) compare for a stem \(u'\), that is \(v\leqF_{u'} w\), then they also compare for every stem \(u\) subsuming \(u'\), that is \(v\leqF_{u} w\) if \(u \leqI u'\).

Expectedly, a FORQ needs to satisfy certain properties for \(T_\calA\) to be finite, computable and for \eqref{eq:crux:equivalence} to hold (in particular the left to right direction).
The property of right-monotonicity of FORQs is needed so that we can iteratively compute \(T_\calA\) via a fixpoint computation (see Section~\ref{sec:forq}).   
\begin{definition}[Suitable \FORQ]
  A \FORQ \(\calF \defeq \tuple{\mathord{\leqI}, \{\mathord{\leqF_{u}}\}_{u \in \Sigma^{*}}}\) is said to be \emph{finite} (resp.\ \emph{decidable}) when \(\mathord{\leqI}\), its converse \(\mathord{\leqI}^{-1}\), and \(\{\mathord{\leqF_u}\}\) for all \(u \in \Sigma^*\) are all well-quasiorders (resp.\ computable).
Given \(L\subseteq\Sigma^\omega\), \(\calF\) is said to \emph{preserve} \(L\) when for all \(u,\hat{u}\in\Sigma^*\) and all \(v, \hat{v} \in \Sigma^+\) if \(\lasso{u}{v}\in L\), \(u\leqI \hat{u}\), \(v\leqF_{\hat{u}} \hat{v}\) and \(\hat{u}\, \hat{v}\leqI \hat{u}\) then \(\lasso{\hat{u}}{\hat{v}}\in L\).
Finally, \(\calF\) is said to be \emph{\(L\)-suitable} (for inclusion) if it is finite, decidable and preserves~\(L\). 
\end{definition}

Intuitively, the ``well'' property on the quasiorders ensures finiteness of \(T_\calA\).
The perservation property ensures completeness: a counterexample to \(L(\calA)\subseteq M\) can only be discarded (that is, not included in \(T_\calA\)) if it is subsumed by another ultimately periodic word in \(T_\calA\) that is also a counterexample to inclusion.

Before defining \(T_\calA\) we introduce for each state \(p\in P\) the sets of words 
\[
  \stem_p \defeq \{ u\in\Sigma^* \mid \calA\colon p_I \run{u} p\}\qquad\text{and}\qquad
  \per_p \defeq \{ v\in\Sigma^+ \mid \calA\colon p \run{v} p\} \enspace .
\]
The set \(\stem_p\) is the set of stems of \(L(\calA)\) that reach state \(p\) in \(\calA\) while the set \(\per_p\) is the set of periods read by a cycle of \(\calA\) on state \(p\).

Given a \(M\)-suitable \FORQ \(\calF \defeq \tuple{\mathord{\leqI}, \{\mathord{\leqF_{u}}\}_{u \in \Sigma^{*}}}\), we let
\begin{equation}\tag{\(\ddagger\)}\label{eq:crux:kernel}
	T_\calA \defeq \big\{ \lasso{u}{v} \st \exists s\in F_\calA\colon u\in U_s,  v\in V_s^w \text{ for some \(w \in W_s\) with \(u\leqI w\)}\big\}
\end{equation}
where for all \(p\in P\), the set \(U_p\) is a basis of \(\stem_p\) with respect to \(\mathord{\leqI}\), that is \(\base[\leqI]{U_p}{\stem_p}\) holds. Moreover \(\base[\leqI^{-1}]{W_p}{\stem_p}\) holds and \(\base[\leqF_w]{V_p^w}{\per_p}\) holds for all \(w\in W_p\).
Note that the quasiorder \(\leqF_{w}\) used to prune the periods of \(\per_p\) depends on a maximal w.r.t. \(\leqI\) stem \(w\) of \(\stem_p\) since \(w\) belongs to the basis \(W_p\) for \(\leqI^{-1}\).
The correctness argument for choosing \(\leqF_{w}\) essentially relies on the FORQ constraint as the proof of \eqref{eq:crux:equivalence} given below shows.
In Section~\ref{sec:conclusion} we will show, that when \(w\) is not ``maximal'' the quasiorder \(\leqF_{w}\) yields a set \(T_\calA\) for which \eqref{eq:crux:equivalence} does not hold.

Furthermore, we conclude from the finite basis property of the quasiorders of \(\calF\) that  \(U_p\), \(W_p\) and \(\{V_p^w\}_{w\in \Sigma^*}\) are finite for all \(p\in P\), hence
\(T_\calA\) is a finite subset of ultimately periodic words of \(L(\calA)\).
Next we prove the equivalence \eqref{eq:crux:equivalence}.
The proof crucially relies on the preservation property of \(\calF\) which allows discarding candidate counterexamples without loosing completeness, that is, if inclusion does not hold a counterexample will be returned.

\begin{proof}[of \eqref{eq:crux:equivalence}]
Consider \(\UltimA \defeq \{\lasso{u}{v} \st  \exists s\in F_{\calA}\colon u\in \stem_s, v\in \per_s, uv \leqI u\}\).
It is easy to show that \(\UltimA = \{\lasso{u}{v} \st  \exists s\in F_{\calA}\colon u\in \stem_s, v\in \per_s\}\) (same definition as \UltimA buth without the constraint \(uv \leqI u\)) by reasoning on properties of well-quasi orders.%
\footnote{%
The case \(\mathord{\subseteq}\) is trivial.
For the case \(\mathord{\supseteq}\), let \(\lasso{u}{v}\) with \(u\in\stem_s\) and \(v\in\per_s\).
If \(uv \leqI u\) then we are done for otherwise consider the sequence \(\{uv^i\}_{i \in \N}\).
Since \(\mathord{\leqI^{-1}}\) is a well-quasiorder, there exists \(x, y \in \N\) such that \(x < y\) and \(uv^x \leqI^{-1}  uv^y\) (viz. \(uv^y \leqI uv^x\)).
Therefore we have  \(\lasso{(uv^x)}{(v^{y-x})} = \lasso{u}{v}\), \((uv^x)\in\stem_s\), \((v^{y-x})\in\per_s\), and  \((uv^x) (v^{y-x}) \leqI (uv^x)\), hence \(\lasso{u}{v}\in\UltimA\).%
} %
It is well-known that \(\omega\)-regular language inclusion holds if and only if it holds for ultimately periodic words.
Formally \(L(\calA) \subseteq M\) holds if and only if \(\UltimA \subseteq M\) holds.
Therefore, to prove \eqref{eq:crux:equivalence}, we show that \(T_\calA \subseteq M \Leftrightarrow \UltimA \subseteq M\).

To prove the implication \(\UltimA \subseteq M \Rightarrow T_\calA \subseteq M\) we start by taking a word \(\lasso{u}{v} \in T_\calA\) such that, by definition~\eqref{eq:crux:kernel}, \(u \in U_s\) and \(v\in V_s^w\) for some \(s\in F_\calA\) and \(w\in W_s\).
We conclude from \(\base[\leqI]{U_s}{\stem_s}\) and \(\base[\leqF_w]{V_s^w}{\per_s}\) that \(u\in U_s \subseteq \stem_s\) and \(v\in V_s^w \subseteq \per_s\).
Thus, we find that \( \lasso{u}{v}\in \UltimA\) hence the assumption \(\UltimA\subseteq M\) shows that \(\lasso{u}{v} \in M\) which proves the implication.

Next, we prove that  \(T_\calA \subseteq M \Rightarrow \UltimA \subseteq M\) holds as well.
Let \(\lasso{u}{v} \in \UltimA\), i.e., such that there exists \(s\in F_\calA\) for which \(u \in \stem_s\) and \(v \in \per_s\), satisfying \(uv \leqI u\).
Since \(u \in \stem_s\) and \(v \in \per_s\), there exist \(u_0\in U_s\), \(w_0 \in W_s\) and \(v_0\in V_s^{w_0}\) such that \(u_0 \leqI u\leqI w_0\) and \(v_0 \leqF_{w_0} v\) thanks to the finite basis property.
By definition we have \(\lasso{u_0}{v_0} \in T_\calA\) and thus we find that \(\lasso{u_0}{v_0} \in M\) since \(T_\calA\subseteq M\).
Next since \(u \leqI w_0\), the \FORQ constraint shows that \(\mathord{\leqF_{w_0}} \subseteq \mathord{\leqF_{u}}\) which, in turn, implies that \(v_0 \leqF_{u} v\) holds.
Finally, we deduce from \(\lasso{u_0}{v_0} \in M\), \(u_0\leqI u\), \(v_0 \leqF_{u} v\), \(uv\leqI u\) and the preservation of \(M\) by the \FORQ \(\calF\) that \(\lasso{u}{v}\in M\).
We thus obtain that \(\UltimA \subseteq M\) and we are done.\qed
\end{proof}

\begin{example}\label{ex:running}
To gain more insights about our approach consider the \BAs of \figurename~\ref{fig:example:automata} for which we want to check whether \(L(\calA)\subseteq L(\calB)\) holds.
From the description of \(\calA\) it is routine to check that \(\stem_{p_I} = \Sigma^* \) and \(\per_{p_I} = \Sigma^+\).
Let us assume the existence\footnote{The definition of the orderings, needed to compute the bases, are given in Example~\ref{example:struct}.} of \(\mathord{\leqI}\) (hence \(\mathord{\leqI}^{-1}\)), \(\mathord{\leqF_\varepsilon}\) and \(\mathord{\leqF_{aa}}\) such that \(a\leqI aa\) holds and so does \(\base[\leqI]{\{\varepsilon,a\}}{\Sigma^*}\), \(\base[\leqI^{-1}]{\{\varepsilon,aa\}}{\Sigma^*}\),  \(\base[\leqF_\varepsilon]{\{b\}}{\Sigma^+}\) and \(\base[\leqF_{aa}]{\{a\}}{\Sigma^+}\).
In addition, we set \(U_{p_I}=\{\varepsilon,a\}\) since \(\base[\leqI]{\{\varepsilon,a\}}{\Sigma^*}\) and \(W_{p_I}=\{\varepsilon,aa\}\) since \(\base[\leqI^{-1}]{\{\varepsilon,aa\}}{\Sigma^*}\).
Moreover \(V_{p_I}^\varepsilon=\{b\}\) since \(\base[\leqF_\varepsilon]{\{b\}}{\Sigma^+}\), and \(V_{p_I}^{aa}=\{a\}\) since \(\base[\leqF_{aa}]{\{a\}}{\Sigma^+}\).
Next by definition~\eqref{eq:crux:kernel} of \(T_\calA\) and from \(a\leqI aa\) we deduce that \(T_\calA = \{ \lasso{\varepsilon}{(b)}, \lasso{a}{(a)} \}\).
Finally, we conclude from \eqref{eq:crux:equivalence} and \( a^\omega\in T_\calA\) that \(a^\omega\in L(\calA)\) (since \(T_\calA \subseteq L(\calA)\)) hence that \(L(\calA) \nsubseteq L(\calB)\) because \(a^\omega \notin L(\calB)\).
By checking membership of the two  ultimately periodic words of \(T_\calA\) into \(L(\calB)\) we thus have shown that \(L(\calA) \subseteq L(\calB)\) does not hold.
\end{example}

In the example above we did not detail how the \FORQ was obtained let alone how to compute the finite bases.
We fill that gap in the next two sections: we define \FORQs based on the underlying structure of a given \BA in Section~\ref{sec:state} and show they are suitable; and we give an effective computation of the bases hence our \FORQ-based inclusion algorithm in Section~\ref{sec:forq}.

\section{Defining \FORQs from the Structure of an Automaton}\label{sec:state}

In this section we introduce a type of \FORQs called structural \FORQs such that
given a \BA \(\calB\) the structural \FORQ induced by \(\calB\) is \(L(\calB)\)-suitable.

\begin{definition}\label{def:struct}
	Let \(\calB \defeq (Q, q_I, \Delta_\calB, F_\calB)\) be a \BA.
	The structural \FORQ of \calB is the pair \(\tuple{\mathord{\leqI^{\calB}}, \{\mathord{\leqF_u^{\calB}}\}_{u \in \Sigma^{*}}}\) where the quasiorders are defined by:
	\begin{align*}
			u_1 \leqI^{\calB} u_2 &\defrel \runIB{u_1} \subseteq \runIB{u_2}
		\\
			v_1 \leqF^{\calB}_{u} v_2 &\defrel \runFB[{\runIB{u}}]{v_1} \subseteq \runFB[{\runIB{u}}]{v_2}
	\end{align*}
	with \(\textRunIB \colon \super{Q} \times \Sigma^* \rightarrow \super{Q}\) and \(\textRunFB \colon \super{Q} \times \Sigma^* \rightarrow \super{Q^2 \times \{\zero, \one\}}\)
	\begin{align*}
			\runIB{u} &\defeq \{q' \in Q \st \calB \colon q_I \run{u} q'\}
		\\
			\runFB[X]{v} &\defeq \{ (q, q', k) \st q \in X, \calB \colon q \run{v} q', (k = \one \Rightarrow \calB \colon q \run{v}_F q')\}
	\end{align*}
\end{definition}

Given \(u \in \Sigma^*\), the set \(\runIB{u}\) contains states that \(u\) can ``target'' from the initial state \(q_I\).
A ``context'' \((q, q', k)\) returned by \textRunFB, consists in a source state \(q \in Q\), a sink state \(q' \in Q\) and a boolean \(k \in \{\top, \bot\}\) that keeps track whether an accepting state is visited.
Note that, having \(\zero\) as last component of a context does \emph{not} mean that no accepting state is visited.
When it is clear from the context, we often omit the subscript \(\calB\) from \textRunIB and \textRunFB. 
Analogously, we omit the \BA from the structural \FORQ quasiorders when there is no ambiguity.
\begin{lemma}
  Given a \BA \(\calB\), the pair \(\tuple{\mathord{\leqI^{\calB}}, \{\mathord{\leqF_u^{\calB}}\}_{u \in \Sigma^{*}}}\) of Definition~\ref{def:struct} is a \FORQ.
\end{lemma}
\begin{proof}
  Let \(\calB \defeq (Q, q_I, \Delta_\calB, F_\calB)\) be a \BA, we start by proving that the \FORQ constraint holds: \( u \leqI^{\calB} u' \implies \mathord{\leqF^{\calB}_{u'}} \subseteq \mathord{\leqF^{\calB}_{u}}\).
	First, we observe that, for all \(Y \subseteq X \subseteq Q\) and all \(v, v' \in \Sigma^*\), we have that \(\runF[{X}]{v} \subseteq \runF[{X}]{v'} \Rightarrow \runF[Y]{v} \subseteq \runF[Y]{v'}\).
	Consider \(u, u' \in \Sigma^*\) such that \(u \leqI^\calB u'\) and \(v, v' \in \Sigma^*\) such that \(v \leqF^\calB_{u'} v'\).
	Let \(X \defeq \runI{u}\) and \(X' \defeq \runI{u'}\), we have that \(X\subseteq X'\) following \(u \leqI^\calB u'\). 
  Next, we conclude from \(v \leqF^\calB_{u'} v'\) that \(\runF[X']{v} \subseteq \runF[X']{v'}\), hence that \(\runF[X]{v} \subseteq \runF[X]{v'}\) by the above reasoning using \(X\subseteq X'\), and finally that \(v \leqF^\calB_u v'\).

  For the right monotonicity, Definition~\ref{def:struct} shows that
  if \(\runI{u} \subseteq \runI{v}\) then \(\runI{ua} \subseteq \runI{va}\), hence we have \(u\leqI v\) implies \(ua\leqI va\) for all \(a\in\Sigma\). 
  The reasoning with the other quasiorders and \textRunF proceeds analogously.
\qed\end{proof}

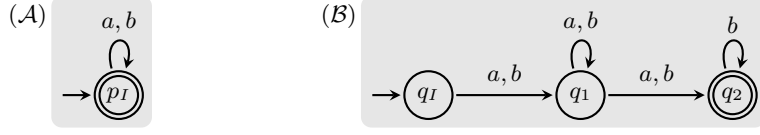
\begin{figure}[t]
	\centering
	\begin{minipage}{.29\linewidth}\centering
		\begin{tikzpicture}[bg={\((\calA)\)},thick]
			\tikzstyle{every state}=[scale=0.75,fill=blue!5,draw=blue!80]
			\node[label=center:\(p_I\),state, initial, final] (i) {};
			\path[transition]
			(i) edge[loop above, above] node{\(a, b\)} (i)
			;
		\end{tikzpicture}
	\end{minipage}
	\begin{minipage}{.69\linewidth}\centering
		\begin{tikzpicture}[bg={\((\calB)\)}, node distance=2cm,thick]
			\tikzstyle{every state}=[scale=0.75,fill=blue!5,draw=blue!80]
			\node[label=center:\(q_I\), state, initial] (i) {};
			\node[label=center:\(q_1\), right of=i, state] (q) {};
			\node[label=center:\(q_2\), right of=q, state, final] (p) {};
			\path[transition]
			(i) edge[above] node{\(a,b\)} (q)
			(q) edge[above] node{\(a,b\)} (p)
			(q) edge[loop above] node{\(a,b\)} (q)
			(p) edge[loop above] node{\(b\)} (p)
			;
		\end{tikzpicture}
	\end{minipage}
	\caption{B\"uchi automata \(\calA\) and \(\calB\) over the alphabet \(\Sigma= \{ a,b \}\).}
	\label{fig:example:automata}
\end{figure}

\begin{example}\label{example:struct}
	Consider the \BA \calB of \figurename~\ref{fig:example:automata} and let \(\tuple{\mathord{\leqI}, \{\mathord{\leqF_u}\}_{u \in \Sigma^{*}}}\) be its structural \FORQ.
  More precisely, we have \(\runI{\varepsilon} = \{q_I\}\); \(\runI{a} = \runI{b}  = \{q_1\}\); and \(\runI{u} = \{q_1, q_2\}\) for all \(u \in \Sigma^{\geq 2}\).
  In particular we conclude from \(u_1 \leqI u_2 \defrel \runI{u_1} \subseteq \runI{u_2}\) that \(a \leqI aa\), \(a\leqI b\) and \(b\leqI a\); \(\varepsilon\) and \(a\) are incomparable; and so are \(\varepsilon\) and \(aa\).
	Since \textRunI has only three distinct outputs, the set \(\{\mathord{\leqF_u}\}_{u \in \Sigma^*}\) contains three distinct quasiorders.
	\begin{enumerate}
    \item \(v_1 \leqF_\varepsilon v_2 \defrel  \runF[\{q_I\}]{v_1} \subseteq \runF[{\{q_I\}}]{v_2} \) where
		\begin{itemize}
      \item \(\runF[\{q_I\}]{\varepsilon} = \{(q_I, q_I, \zero)\}\)
			\item \(\runF[\{q_I\}]{a} = \runF[\{q_I\}]{b} = \{(q_I, q_1, \zero)\}\)
			\item \(\runF[\{q_I\}]{v} = \{(q_I, q_1, \zero), (q_I, q_2, \zero), (q_I, q_2, \one)\}\) for all \(v \in\Sigma^{\geq 2}\) 
		\end{itemize}
			\item \(v_1 \leqF_a v_2 \iff v_1 \leqF_b v_2 \defrel \runF[\{q_1\}]{v_1} \subseteq \runF[{\{q_1\}}]{v_2} \) where
		\begin{itemize}
      \item \(\runF[\{q_1\}]{\varepsilon} = \{ (q_1,q_1,\zero) \}\)
      \item \(\runF[\{q_1\}]{v} =  \{(q_1, q_1, \zero), (q_1, q_2, \zero), (q_1, q_2, \one)\}\) for all \(v\in\Sigma^+\)
		\end{itemize}
		\item \(v_1 \leqF_{u_1} v_2 \iff v_1 \leqF_{u_2} v_2  \defrel \runF[\{q_1,q_2\}]{v_1} \subseteq \runF[{\{q_1,q_2\}}]{v_2} \) for all 
\(u_1,u_2\in\Sigma^{{}\geq 2}\) where
		\begin{itemize}
      \item \(\runF[\{q_1, q_2\}]{\varepsilon} = \{(q_1,q_1,\zero),(q_2,q_2,\zero),(q_2,q_2,\one)\}\)
			\item \(\runF[\{q_1, q_2\}]{v} = \{(q_1, q_1, \zero), (q_1, q_2, \zero), (q_1, q_2, \one)\}\) for all \(v \in \Sigma^+\backslash\{b\}^{+}\)
      \item \(\runF[\{q_1, q_2\}]{v} = \{(q_1, q_1, \zero), (q_1, q_2, \zero), (q_1, q_2, \one), (q_2,q_2,\zero),(q_2,q_2,\one)\}\) for all \(v \in \{b\}^{+}\). 
		\end{itemize}
	\end{enumerate}
  With the above definitions the reader is invited to check the following predicates \(\base[\leqI]{\{\varepsilon,a\}}{\Sigma^*}\), \(\base[\leqI]{\{\varepsilon,b\}}{\Sigma^*}\), \(\base[\leqI^{-1}]{\{\varepsilon,aa\}}{\Sigma^*}\), \(\base[\leqF_\varepsilon]{\{b\}}{\Sigma^+}\), \(\base[\leqF_{a}]{\{b\}}{\Sigma^+}\) and \(\base[\leqF_{aa}]{\{a\}}{\Sigma^+}\). 
  Also observe that none of the above finite bases contains comparable words for the ordering thereof.
  We also encourage the reader to revisit Example~\ref{ex:running}.
\end{example}

As prescribed in Section~\ref{sec:crux}, we show that for every \BA \(\calB\) its structural \FORQ is \(L(\calB)\)-suitable, namely it is finite, decidable and preserves \(L(\calB)\). 

\begin{proposition}\label{prop:struct:preservation}
	Given a \BA \calB, its structural \FORQ is \(L(\calB)\)-suitable.
\end{proposition}
\begin{proof}
	Let  \(\calB \defeq (Q, q_I, \Delta_\calB, F_\calB)\) be a \BA and \(\calF \defeq \tuple{\leqI, \{\leqF_{u}\}_{u\in \Sigma^*}}\) be its structural \FORQ.
  The finiteness proof of \(\calF\) is trivial since \(Q\) is finite and so is the proof of decidability by Definition~\ref{def:struct}. 
  For the preservation, given \( \lasso{u_0}{v_0} \in L(\calB)\), we show that for all \(u \in \Sigma^*\) and all \(v \in \Sigma^+\) such that \(uv \leqI u\) and \(u_0 \leqI u \) and \(v_0 \leqF_{u} v\) then \(\lasso{u}{v} \in L(\calB)\) holds.
	Let a run \(\pi_0 \defeq q_I \run{u_0} q_0 \run{v_0} q_1 \run{v_0} q_2 \dots \) of \calB over \(\lasso{u_0}{v_0}\) which is accepting.
  Stated equivalently, we have \(q_0\in\runI{u_0}\) and 
  \( (q_i,q_{i+1},x_i) \in \runF[\runI{u_0v_0^i}]{v_0} \) for every \(i\in\N\) with the additional constraint that \(x_i = \top\) holds infinitely often. 
	
	We will show that \calB has an accepting run over \(\lasso{u}{v}\) by showing that \(q_0\in\runI{u}\) holds; \( (q_i,q_{i+1},x_i) \in \runF[\runI{uv^i}]{v} \) holds for every \(i\in\N\); and \(x_i = \top\) holds infinitely often.
  Since \(u_0 \leqI u\) and \(q_0 \in \runI{u_0}\) we find that \(q_0 \in \runI{u}\) by definition of \(\mathord{\leqI}\).
	Next we show the remaining constraints by induction.
	The induction hypothesis states that for all \(0\leq n\) we have \((q_{n},q_{n+1},x_n)\in\runF[{\runI{uv^{n}}}]{v}\).
	For the base case (\(n=0\)) we have to show that
	\( (q_0,q_1,x_0) \in \runF[{\runI{u}}]{v}\).
	We conclude from \( (q_0,q_1,x_0) \in \runF[{\runI{u}}]{v_0}\), \( v_0 \leqF_{u} v\) and the definition of \(\leqF_{u}\) that \(\runF[{\runI{u}}]{v_0}\subseteq \runF[{\runI{u}}]{v}\) and finally that \( (q_0,q_1,x_0) \in  \runF[{\runI{u}}]{v}\).
	For the inductive case, assume \((q_{n},q_{n+1},x_n)\in\runF[{\runI{uv^{n}}}]{v}\).
  The definition of context shows that \(q_{n+1} \in \runI{uv^{n+1}}\).
	It takes an easy an induction to show that \(uv^n \leqI u\) for all \(n\) using \(uv\leqI u\)  and right-monotonicity of \(\mathord{\leqI}\).
  We conclude from \(uv^{n+1} \leqI u\), the definition of \(\mathord{\leqI}\) and \(q_{n+1} \in \runI{uv^{n+1}}\) that \(q_{n+1}\in\runI{u}\) also holds, hence that
  \( (q_{n+1}, q_{n+2}, x_{n+1}) \in \runF[\runI{u}]{v_0}\) following the definition of contexts and that of \(\pi_0\).
  Next, we find that \( (q_{n+1},q_{n+2},x_{n+1})\in \runF[{\runI{u}}]{v}\) following a reasoning analogous to the base case, this time starting with \( (q_{n+1}, q_{n+2}, x_{n+1}) \in \runF[\runI{u}]{v_0})\).
  Finally, \(q_{n+1} \in \runI{uv^{n+1}}\) implies that \((q_{n+1},q_{n+2},x_{n+1})\in \runF[{\runI{uv^{n+1}}}]{v}\).
  We have thus shown that \(q_0\in\runI{u}\) and \( (q_i,q_{i+1},x_i) \in \runF[\runI{uv^i}]{v} \) for every \(i\in\N\) with the additional constraint that \(x_i = \top\) holds infinitely often and we are done.\qed
\end{proof}

\section{A \FORQ-based Inclusion Algorithm}\label{sec:forq}

As announced at the end of Section~\ref{sec:crux} it remains, in order to formulate our \FORQ-based algorithm deciding whether \(L(\calA) \subseteq M\) holds, to give an effective computation for the bases defining \(T_\calA\).
We start with a fixpoint characterization of the stems and periods of \BAs using
the function \(\textRcatA \colon \super{\Sigma^*}^{|P|} \rightarrow \super{\Sigma^*}^{|P|}\):
\begin{equation*}
  \textRcatA(\vec{X})\at{p} \defeq \vec{X}\at{p} \cup \big\{ wa \in \Sigma^* \st w \in \vec{X}\at{p'}, a \in \Sigma, \calA \colon p' \run{a} p \big\}
\end{equation*}
where \(\vec{S}\at{p}\) denotes the \(p\)-th element of the vector \(\vec{S} \in \super{\Sigma^*}^{|P|}\).
In \figurename~\ref{algo:forq}, the repeat/until loops at lines~\ref{algo:forq:line:loopW} and \ref{algo:forq:line:loopU} compute iteratively subsets of the stems of \calA, while the loop at line~\ref{algo:forq:line:loopV} computes iteratively subsets of the periods of \calA.
The following lemma formalizes the above intuition. 

\begin{lemma}\label{lemma:rcat}
Consider \(\vec{U}_0\) and \(\vec{V}_1^s\) (with \(s\in F_\calA\)) in the \FORQ-based algorithm.
The following holds for all \(n \in \N\):
	\[
		\catA[n]{\vec{U}_0}\at{p} = \stem_p \cap \Sigma^{\leq n}\text{ for all \(p\in P\), and }
		\catA[n]{\vec{V}_1^s}\at{s} = \per_s \cap \Sigma^{\leq n+1}\enspace .
	\]
\end{lemma}

\begin{figure}[t]\centering
	\begin{tikzpicture}[bg]
		\node[text width=\linewidth-1.5em]{\hspace{-0.6em}
			\begin{minipage}{\linewidth+2.4em}
				\begin{algorithm}[H]
					\DontPrintSemicolon
					\SetInd{.2em}{.6em}
					\SetKwIF{If}{ElseIf}{Else}{if}{then}{else if}{else}{}
					\SetKwFor{For}{for each}{do}{}
					\SetKwIF{DefIf}{DefElseIf}{DefElse}{}{with}{and}{and}{}
					\SetKwFor{Def}{let}{as}{}
					\SetKwFor{Select}{select}{such that}{}
					\SetKwFor{While}{while}{do}{}
					\SetKwRepeat{Repeat}{repeat}{until}
					\SetKw{Return}{return}
					\SetKw{And}{and}
					\SetKw{St}{such that}
					\SetKwProg{FUNCTON}{Function}{:}{}
					\SetKwInput{IN}{Input}
					\SetKwInput{OUT}{Output}
					\SetKwInput{DATA}{Data}
					\SetKwComment{ccc}{\color{gray}/* }{\color{gray}\ */\quad}

					\IN{B\"uchi automaton \(\calA \defeq (P, p_I, \Delta_\calA, F_\calA)\)}
					\IN{\(\omega\)-regular language \(M\) with  procedure deciding \(\lasso{u}{v}\in M\) given \(u,v\)}
					\IN{\(M\)-suitable \FORQ \(\calF \defeq \tuple{\mathord{\leqI}, \{\mathord{\leqF_u}\}_{u\in \Sigma^* }}\)}
					\OUT{Returns \ok if \(L(\calA)\subseteq M\) and \ko otherwise}
					\BlankLine
					
					\SetAlgoNoLine
					\FUNCTON{}{
						\SetAlgoVlined
						\lDef*
						{\(\vec{U}_0 \in \super{\Sigma^*}^{|P|}\)}
						{\leDefIf{\(\vec{U}_0\at{p} \defeq \varnothing\)}{\(p \neq p_I\)}{\(\vec{U}_0\at{p_I} \defeq \{\varepsilon\}\)}}\nllabel{algo:forq:line:U0}
						\(\vec{W} \coloneqq \vec{U} \coloneqq \vec{U}_0\)\nllabel{algo:forq:line:Winit}\;
						
						\lRepeat
						{\(\base[\mathord{\leqI^{-1}}]{\vec{W}\at{p}}{\catA{\vec{W}}\at{p}}\) for all \(p\in P\)}%
            {\(\vec{W} \coloneqq \catA{\vec{W}}\)} \nllabel{algo:forq:line:loopW}
						\lRepeat
						{\(\base[\mathord{\leqI}]{\vec{U}\at{p}}{\catA{\vec{U}}\at{p}}\) for all \(p\in P\)}%
            {\(\vec{U} \coloneqq \catA{\vec{U}}\)} \nllabel{algo:forq:line:loopU}
						
						\For{\(s\in F_\calA\)}{ \nllabel{algo:forq:line:final}
							\lDef*
							{\(\vec{V}_1^s \in \super{\Sigma^*}^{|P|}\)}
							{\lDefIf{\(\vec{V}_1^s\at{p} \defeq \{ a \in \Sigma \st \calA \colon s \run{a} p\}\)}{\(p \in P\)}}\nllabel{algo:forq:line:V1}
							
							\For{\(w \in \vec{W}\at{s}\)}{ \nllabel{algo:forq:line:W}
								\(\vec{V}^s \coloneqq \vec{V}_1^s\)\nllabel{algo:forq:line:Vsinit}\;
								\lRepeat
								{\(\base[\mathord{\leqF_w}]{\vec{V}^s\at{p}}{\catA{\vec{V}^s}\at{p}}\) for all \(p\in P\)}
								{\(\vec{V}^s \coloneqq \catA{\vec{V}^s}\)} \nllabel{algo:forq:line:loopV}
								
								\For{\(v\in \vec{V}^s\at{s}\)}{ \nllabel{algo:forq:line:V}
									\For{\(u\in \vec{U}\at{s}\) \St \(u \leqI w \)}{\nllabel{algo:forq:line:U}
										\lIf{\(\lasso{u}{v}\notin M\)}{\Return \ko}}\nllabel{algo:forq:line:ko}
								}
							}
						}\Return \ok}
          \end{algorithm}
        \end{minipage}};
	\end{tikzpicture}
	\caption{\FORQ-based algorithm}\label{algo:forq}
\end{figure}

Prior to proving the correctness of the algorithm of \figurename~\ref{algo:forq} we need the following result which is key for establishing the correctness of the repeat/until loop conditions of lines~\ref{algo:forq:line:loopW}, \ref{algo:forq:line:loopU}, and~\ref{algo:forq:line:loopV}.
\begin{lemma}\label{lem:rightmono}
	 Let \(\mathord{\tmprel}\) be a right-monotonic quasiorder over \(\Sigma^*\).
	 Given \linebreak\(\calA \defeq (P, p_I, \Delta_\calA, F_\calA)\) and \(\vec{S}, \vec{S'} \in \super{\Sigma^*}^{|P|}\), if \(\base[\mathord{\tmprel}]{\vec{S'}\at{p}}{\vec{S}\at{p}}\) holds for all \(p \in P\) then \(\base[\mathord{\tmprel}]{\catA{\vec{S'}}\at{p}}{\catA{\vec{S}}\at{p}}\) holds for all \(p \in P\).
\end{lemma}
\begin{proof}
	Consider \(w\in\catA{\vec{S}}\at{p}\) where \(p\in P\), we show that there exists \(w'\in\catA{\vec{S}'}\at{p}\) such that \(w' \tmprel w\).
	Assume that \(\base[\mathord{\tmprel}]{\vec{S}'\at{p}}{\vec{S}\at{p}}\) holds for all \(p\in P\).
	In particular, for all \(w_1 \in \vec{S}\at{p}\), there exists \(w'_1 \in \vec{S}'\at{p}\) such that \(w'_1 \tmprel w_1\).
	In the case where \(w_1 \in \catA{\vec{S}}\at{p} \setminus \vec{S}\at{p}\), by definition of \textRcatA \(w_1\) is of the form \(w_2a\) for some \(a\in\Sigma\) and some \(w_2\in \vec{S}\at{\hat{p}}\) such that \(\calA \colon \hat{p} \run{a} p\).
	Since \(\base[\mathord{\tmprel}]{\vec{S}'\at{\hat{p}}}{\vec{S}\at{\hat{p}}}\) and \(w_2\in \vec{S}\at{\hat{p}}\), there exists \(w_3\in \vec{S}'\at{\hat{p}}\) such that \(w_3 \tmprel w_2\).
	We deduce that \(w_3 a \tmprel w_2a\) holds, hence \(w_3 a \tmprel w_1\) holds as well from the right-monotonicity of \(\mathord{\tmprel}\).
	Furthermore \(w_3a \in\catA{\vec{S}'}\at{p}\) by definition of \textRcatA and since \(\calA \colon \hat{p} \run{a} p\). 
	Finally, we conclude that \(\base[\mathord{\tmprel}]{\catA{\vec{S'}}}{\catA{\vec{S}}}\) holds.
	\qed
\end{proof}

\begin{theorem}
  The \FORQ-based algorithm decides the inclusion of \BAs.
\end{theorem}
\begin{proof}
  We first show that every loop of the algorithm eventually terminates.
  First, we conclude from the definition of \textRcatA and the initializations (lines~\ref{algo:forq:line:Winit} and \ref{algo:forq:line:Vsinit}) of each repeat/until loop (lines~\ref{algo:forq:line:loopW}, \ref{algo:forq:line:loopU}, and~\ref{algo:forq:line:loopV}) that each component of each vector holds a finite set of words.
  Observe that the halting conditions of the repeat/until loops are effectively computable since every quasiorder of \(\calF\) is decidable and because, in order to decide \(\base[\mathord{\tmprel}]{Y}{X}\) where \(X,Y\) are finite sets and \(\mathord{\tmprel}\) is decidable, it suffices to check that \(Y\subseteq X\) and that for every \(x\in X\) there exists \(y\in Y\) such that \(y \tmprel x\).
  Next, we conclude from the fact that all the quasiorders of \(\calF\) used in the repeat/until loops are all well-quasiorders that there is no infinite sequence \(\{X_i\}_{i\in\N}\) such that \(\up{\mathord{\tmprel}}{X_1} \subsetneq \up{\mathord{\tmprel}}{X_2} \subsetneq \ldots\)
  Since \(\base[\mathord{\tmprel}]{Y}{X}\) is equivalent to \(Y\subseteq X\land \up{\mathord{\tmprel}}{X}\subseteq\up{\mathord{\tmprel}}{Y}\) and since each time \textRcatA updates a component its upward closure after the update includes the one before, we find that every repeat/until loop must terminate after finitely many iterations.

  Next, we show that when the repeat/until loop of line~\ref{algo:forq:line:loopU} halts, \(\base[\mathord{\leqI}]{\vec{U}\at{p}}{\stem_p}\) holds for all \(p\in P\).
  It takes an easy induction on \(n\) together with Lemma~\ref{lem:rightmono} to show that if \(\base[\mathord{\leqI}]{\catA[n+1]{\vec{U_0}}\at{p}}{\catA[n]{\vec{U_0}}\at{p}}\) holds for all \(p\in P\) then \(\base[\leqI]{\catA[n]{\vec{U_0}}\at{p}}{\catA[m]{\vec{U_0}}\at{p}}\) holds for all \(m>n\).
  Hence Lemma~\ref{lemma:rcat} shows that \(\base[\mathord{\leqI}]{\catA[k]{\vec{U_0}}\at{p}}{\stem_p}\) holds for all \(p\in P\) where \(k\) is the number of iterations of the repeat/until loop implying \(\base[\mathord{\leqI}]{\vec{U}\at{p}}{\stem_p}\) holds when the loop of line~\ref{algo:forq:line:loopU} halts.

  An analogue reasoning shows that \(\base[\mathord{\leqI^{-1}}]{\vec{W}\at{p}}{\stem_p}\) holds for all \(p\in P\), as well as \(\base[\mathord{\leqF_{w}}]{\vec{V}^s\at{s}}{\per_s}\) holds for all \(w\in \vec{W}\at{s}\) and all \(s\in F_\calA\) upon termination of the loops of lines~\ref{algo:forq:line:loopW} and~\ref{algo:forq:line:loopV}.

  To conclude, we observe that each time a membership query is performed at line~\ref{algo:forq:line:ko}, the ultimately periodic word \(\lasso{u}{v}\) belongs to \(T_\calA\) defined by (\ref{eq:crux:kernel}).
  This is ensured since \(u\in\base[\mathord{\leqI}]{\vec{U}\at{s}}{\stem_s}\), \(w\in\base[\mathord{\leqI^{-1}}]{\vec{W}\at{s}}{\stem_s}\), \(v\in \base[\mathord{\leqF_{w}}]{\vec{V}^s\at{s}}{\per_s}\) for some \(s \in F_\calA\) and, thanks to the test at line~\ref{algo:forq:line:U}, the comparison \(u\leqI w\) holds.
\qed
\end{proof}

\begin{remark}\label{rmk:finitebasis}
	The correctness of the \FORQ-based algorithm still holds when, after every ``\(\coloneqq\)'' assignment (at lines~\ref{algo:forq:line:Winit}, \ref{algo:forq:line:loopW}, \ref{algo:forq:line:loopU}, \ref{algo:forq:line:Vsinit} and~\ref{algo:forq:line:loopV}), we remove from the variable content zero or more subsumed words for the corresponding ordering.
	The effect of removing zero or more subsumed words from a variable can be achieved by replacing assignments like, for instance, \(\vec{U} \coloneqq \catA{\vec{U}}\) at line~\ref{algo:forq:line:loopU} with \(\vec{U} \coloneqq \catA{\vec{U}}; \vec{U}\coloneqq \vec{U}_r\) where \(\vec{U}_r\) satisfies \(\base[\mathord{\leqI}]{\vec{U}_r\at{p}}{\vec{U}\at{p}}\) for all \(p\in P\).
  The correctness of the previous modification follows from Lemma~\ref{lem:rightmono}.
  Therefore, the sets obtained by discarding subsumed words during computations still satisfy the basis predicates of \(T_\calA\) given at~\eqref{eq:crux:kernel}.
\end{remark}

It is worth pointing out that the correctness arguments developed above, do not depend on the specifics of the structural \FORQs.
The \FORQ-based algorithm is sound as long as we provide a suitable \FORQ.
Next we study the algorithmic complexity of the algorithm of \figurename~\ref{algo:forq}.

\section{Complexity of the structural \FORQ-based algorithm}\label{sec:complexity} %
In this Section, we establish an upper bound on the runtime of the algorithm of \figurename~\ref{algo:forq} when the input \FORQ is the structural \FORQ induced by a \BA \(\calB\).
Let \(n_\calA\) and \(n_\calB\) be respectively the number of states in the \BA \(\calA\) and \(\calB\).
We start by bounding the number of iterations in the repeat/until loops.
In each repeat/until loop, each component of the vector holds a finite set of words the upward closure of which  grows (for \(\subseteq\)) over time and when all the upward closures stabilize the loop terminates.
In the worst case, an iteration of the repeat/until loop adds exactly one word to some component of the vector which keeps the halting condition falsified (the upward closure strictly increases).
Therefore a component of the vector cannot be updated more than \(2^{n_\calB}\) times for otherwise its upward closure has stabilized.
We thus find that the total number of iterations is bounded from above by \(n_\calA \cdot 2^{n_\calB}\) for the loops computing \(\vec{U}\) and \(\vec{W}\).
Using an analogous reasoning we conclude that each component of the \(\vec{V}\) vector has no more than \(2^{(2 {n_\calB}^2)}\) elements and  the total number of iterations is upper-bounded by \( n_\calA \cdot 2^{(2 {n_\calB}^2)}\).
To infer an upper bound on the runtime of each repeat/until loop we also need to multiply the above expressions by a factor \(|\Sigma|\) since the number of concatenations in \textRcat depends on the size of the alphabet.

Next, we derive an upper bound on the number of membership queries performed at line~\ref{algo:forq:line:ko}.
The number of iterations of the loops of lines \ref{algo:forq:line:final}, \ref{algo:forq:line:W}, \ref{algo:forq:line:loopV}, \ref{algo:forq:line:V} and \ref{algo:forq:line:U} is \(n_\calA\), \(2^{n_\calB}\), \( n_\calA \cdot 2^{(2 {n_\calB}^2)}\), \( 2^{(2 {n_\calB}^2)}\) and \(2^{n_\calB}\), respectively.
Since all loops are nested, we multiply these bounds to end up with \(n_\calA^2 \cdot 2^{\mathcal{O}(n_\calB^2)}\) as an upper bound on the number of membership queries.
The runtime for each ultimately periodic word membership query (with a stem, a period and \(\calB\) as input) is upper bounded by an expression polynomial in the size \(n_\calB\) of \(\calB\), \(2^{n_\calB}\) for the length of the stem and \(2^{(2 {n_\calB}^2)}\) for the length of the period.

We conclude from the above that the runtime of the algorithm of \figurename~\ref{algo:forq} is at most \(|\Sigma| \cdot n_\calA^2 \cdot 2^{\mathcal{O}(n_\calB^2)}\).

\section{Implementation and experiments}\label{sec:implementation}

We implemented the \FORQ-based algorithm of \figurename~\ref{algo:forq} instantiated by the structural \FORQ in a tool called \FORKLIFT\cite{forklift-github}.
In this section, we provide algorithmic details about \FORKLIFT and then analyze how it behaves in practice (Section~\ref{sub:evaluation}).

\paragraph{Data structures.}
Comparing two words given a structural \FORQ requires to compute the corresponding sets of target for stems (\textRunI), and sets of context for periods (\textRunF).
A na\"ive implementation would be to compute \textRunI and \textRunF every time a comparison is needed.
We avoid to compute this information over and over again by storing each stem together with its \textRunI set and each period together with its \textRunF set.

Moreover, the function \textRcat inserts new words in the input vector by concatenating a letter on the right to some words already in the vector.
In our implementation, we do not recompute the associated set of targets nor context for the newly computed word from scratch.
For all stem \(u \in \Sigma^*\) and all letter \(a \in \Sigma\), the set of states \(\runI{ua}\) can be computed from  \(\runI{u}\) thanks to the following equality essentially stating that \(\runI{}\) can be computed inductively:
\[
  \runI{ua} = \big\{ q\in Q \st q'\in\runI{u},\, \calB\colon q'\run{a} q \big\}\enspace .
\]
Analogously, for all period \(v \in \Sigma^+\), all  \(X \subseteq Q\) and all \(a \in \Sigma\), the set of contexts \(\runF[X]{va}\) can be computed from \(\runF[X]{v}\) thanks to the following equality:
\[
	\runF[X]{va} =  \bigg\{ (q_0,q,k)\in Q^2\times\{\bot,\top\} ~\Big|~ \begin{array}{l}
		(q_0,q',k')\in  \runF[X]{v},\, \calB\colon q'\run{a}q
		\smallskip\\
    (k=\bot\lor k'=\top \lor \calB\colon q' \run{a}_F q)
	\end{array} \bigg\}\enspace .
\]
Intuitively \textRunF can be computed inductively as we did for \textRunI. 
The first part of the condition defines how new context are obtained by appending a transition to the right of an existing context while the second part defines the bit of information keeping record of whether an accepting state was visited. 

\paragraph{Bases, Frontier and Membership test.}
We stated in Remark~\ref{rmk:finitebasis} that the correctness of the \FORQ-based algorithm is preserved when removing, from the computed sets, zero or more subsumed words for the corresponding ordering.
In \FORKLIFT, we remove all the subsumed words from all the sets we compute which, intuitively, means each computed set is a basis that contains as few words as possible.
To remove subsumed words we leverage the target or context sets kept along with the words.
It is worth pointing out that the least fixpoint computations at lines~\ref{algo:forq:line:loopW}, \ref{algo:forq:line:loopU}, and \ref{algo:forq:line:loopV} are implemented using a frontier.
Finally, the ultimately periodic word membership procedure is implemented as a classical depth-first search as described in textbooks \cite[Chapter~13.1.1]{esparzaAutomataTheoryAlgorithmic2017}.

\paragraph{Technical details.}
\FORKLIFT, a na\"ive prototype implemented by a single person over several weeks, implements the algorithm of \figurename~\ref{algo:forq} with the structural \FORQ in less than 1\,000 lines of Java code.
One of the design goals of our tool was to have simple code that could be easily integrated in other tools.
Therefore, our implementation relies solely on a few standard packages from the Java SE Platform (notably collections such as \texttt{HashSet} or \texttt{HashMap}).

\subsection{Experimental Evaluation}%
\label{sub:evaluation}

\paragraph{Benchmarks.}

Our evaluation uses benchmarks stemming from various application domains including benchmarks from theorem proving, software verification, and from previous work on the \(\omega\)-regular language inclusion problem.
In this section, a \emph{benchmark} means an ordered pair of \BAs such that the “left”/“right” \BAs refer, resp., to the automata on the left/right of the inclusion sign.
The \BAs of the Pecan \cite{DBLP:journals/corr/abs-2102-01727} benchmarks encode sets of solutions of predicates, hence a logical implication between predicates reduces to a language inclusion problem between \BAs.
The benchmarks correspond to theorems of  type \(\forall x, \exists y,\, P(x) \implies Q(y)\) about Sturmian words \cite{ReedCSL22}.
We collected 60 benchmarks from Pecan for which inclusion holds, where the \BAs have alphabets of up to 256 symbols and have up to 21\,395 states.

The second collection of benchmarks stems from software verification.
The Ultimate Automizer (UA)~\cite{DBLP:conf/tacas/HeizmannCDGHLNM18,DBLP:conf/cav/HeizmannHP13} benchmarks encode termination problems for programs where the left \BA models a program and the right \BA its termination proof.
Overall, we collected 600 benchmarks from UA for which inclusion holds for all but one benchmark.
The \BAs have alphabets of up to 13\,173 symbols and are as large as 6\,972 states.

The RABIT benchmarks are \BAs modeling mutual exclusion algorithms~\cite{DBLP:conf/concur/AbdullaCCHHMV11}, where in each benchmark one \BA is the result of translating a set of guarded commands defining the protocol while the other \BA translates a modified set of guarded commands, typically obtained by randomly weakening or strengthening one guard.
The resulting \BAs are on a binary alphabet and are as large as 7\,963 states.
Inclusion holds for 9 out of the 14 benchmarks.

All the benchmarks are publicly available on GitHub~\cite{bait-benchmarks-github}.
We used all the benchmarks we collected, that is, we discarded no benchmarks.

\paragraph{Tools.}%
\label{sub:tools}
We compared \FORKLIFT with the following tools: SPOT 2.10.3, GOAL (20200822), RABIT 2.5.0, ROLL 1.0, and BAIT 0.1. 

\begin{description}
 \item[SPOT] \cite{DBLP:conf/atva/Duret-LutzLFMRX16,alexandreduret-lutzSpotSpot10} decides inclusion problems by complementing the “right” \BA via determinization to parity automata with some additional optimizations including simulation-based optimizations. %
 It is invoked through the command line tool \texttt{autfilt} with the option \texttt{--included-in}.
 It is worth pointing out that SPOT works with symbolic alphabets where symbols are encoded using Boolean propositions, and sets of symbols are represented and processed using OBDDs.
 SPOT is written in {C\nolinebreak[4]\hspace{-.05em}\raisebox{.4ex}{\tiny\bf ++}} and its code is publicly available~\cite{spot-tool}.

 \item[GOAL] \cite{DBLP:conf/cav/TsaiTH13} contains several language inclusion checkers available with multiple options.
 We used the Piterman algorithm using the options \texttt{containment -m piterman} with and without the additional options \texttt{-sim -pre}.
 In our plots \(\text{GOAL}\) is the invocation with the additional options \texttt{-sim -pre} which compute and use simulation relations to further improve performance while \(\text{GOAL}^{-}\) is the one without the additional options.
 Inclusion is checked by constructing on-the-fly the intersection of the “left” \BA and the complement of the “right” \BA which is itself built on-the-fly by the Piterman construction \cite{Piterman2007}.
 The Piterman check was deemed the “best effort” (cf.~\cite[Section~9.1]{DBLP:journals/lmcs/ClementeM19} and \cite{DBLP:journals/corr/TsaiFVT14}) among the inclusion checkers provided in GOAL.
 GOAL is written in Java and the source code of the release we used is not publicly available~\cite{goal-tool}.

 \item[RABIT] \cite{DBLP:journals/lmcs/ClementeM19} performs the following operations to check inclusion:
 (1)~Removing dead states and minimizing the automata with simulation-based techniques, thus yielding a smaller instance;
 (2)~Witnessing inclusion by simulation already during the minimization phase; 
 (3)~Using a Ramsey-based method with antichain heuristics
 to witness inclusion or non-inclusion.
 The antichain heuristics of Step~(3) uses a unique quasiorder leveraging simulation relations to discard candidate counterexamples.
 In our experiments we ran RABIT with options \texttt{-fast -jf} which RABIT states as providing the ``best performance''.
 RABIT is written in Java and is publicly available~\cite{rabit-tool}.
 
 \item[ROLL] \cite{Li2020b,DBLP:conf/tacas/LiSTCX19} contains an inclusion checker that does a preprocessing analogous to that of RABIT and then relies on automata learning and word sampling techniques to decide inclusion.
 ROLL is written in Java and is publicly available~\cite{roll}.

 \item[BAIT] \cite{DBLP:conf/concur/DoveriGPR21}  which shares authors with the authors of the present paper, implements a Ramsey-based algorithm with the antichain heuristics where two quasiorders (one for the stems and the other for the periods) are used to discard candidate counterexamples as described in Section~\ref{sec:introduction}.
 BAIT is written in Java and is publicly available~\cite{bait-github}.
\end{description}

 As far as we can tell all the above implementations, including \FORKLIFT, are sequential except for RABIT which, using the \texttt{-jf} option, performs some computations in a separate thread.

\paragraph{Experimental Setup.}
We ran our experiments on a server with \(24\,\text{GB}\) of RAM, 2 Xeon E5640 \(2.6~\text{GHz}\) CPUs and Debian Stretch \(64\)-bit.
We used openJDK 11.0.12 2021-07-20 when compiling Java code and ran the JVM with default options.
For RABIT, BAIT and \FORKLIFT the execution time is computed using timers internal to their implementations.
For ROLL, GOAL and SPOT the execution time is given by the ``real'' value of the \texttt{time(1)} command.
We preprocessed the benchmarks passed to \FORKLIFT and BAIT  with a reduction of the set of final states of the ``left'' \BA that does not alter the language it recognizes.
This preprocessing aims to minimize the number of iterations of the loop at line~\ref{algo:forq:line:final} of \figurename~\ref{algo:forq} over the set of final states.
It is carried out by GOAL using the \texttt{acc -min} command. 
Internally, GOAL uses a polynomial time algorithm that relies on computing strongly connected components.
The time taken by this preprocessing is negligible.

\paragraph{Plots.}%
\label{sub:plots}

We use survival plots for displaying our experimental results in \figurename~\ref{fig:experiments}.
Let us recall how to obtain them for a family of benchmarks \(\{p_i\}_{i=1}^n\):
(1)~run the tool on each benchmark \(p_i\) and store its runtime \(t_i\);
(2)~sort the \(t_i\)'s in increasing order and discard pairs corresponding to abnormal program termination like time out or memory out;
(3)~plot the points \((t_1,1), (t_1+t_2,2)\),\ldots, and in general \((\sum_{i=1}^k t_i,k)\);
(4)~repeat for each tool under evaluation.

Survival plots are effective at comparing how tools scale up on benchmarks: the further right and the flatter a plot goes, the better the tool thereof scales up. 
Also the closer to the \(x\)-axis a plot is, the less time the tool needs to solve the benchmarks. 

\begin{figure}[b!]
\centering
\subfloat[Benchmarks from Pecan]{
\begin{tikzpicture}[scale=.65]
			\begin{axis}[width=15cm, height=10cm, ylabel=time (ms), ymode=log, yticklabel pos=right, xlabel=\# instances, xmin=41, xmax=60, xtick={41,43,52,54,57,58,59,60}, legend style={legend columns=2, legend pos=north west}, tick style={grid=major}, scaled ticks=false, tick label style={/pgf/number format/fixed}]
				\addplot[mark size=2.5pt, line width=1pt, mark=diamond,color=cyan!85!blue]         file {pecan-goal.dat} ;%
				\addplot[mark size=2.5pt, line width=1pt, mark=diamond*,color=cyan!15!blue]  file {pecan-goal-minus.dat} ;%
				\addplot[mark size=2.5pt, line width=1pt, mark=star,color=brown!80!violet]          file {pecan-rabit.dat} ;%
				\addplot[mark size=2.5pt, line width=1pt, mark=square*,color=green!25!gray]  file {pecan-bainc.dat} ;%
        \addplot[mark size=2.5pt, line width=1pt, mark=Mercedes star, color=yellow!50!orange] file {pecan-spot-2_10_3.dat} ;%
				\addplot[mark size=2.5pt, line width=1pt, mark=*,color=magenta!20!red] file {forklift-last-pecan.dat}; %
				\addplot[mark size=2.5pt, line width=1pt, mark=|,color=black]      file {pecan-roll.dat} ;%
				\legend{GOAL, GOAL\(^{-}\), RABIT, BAIT, SPOT, \FORKLIFT, ROLL}
			\end{axis}
		\end{tikzpicture}
\label{fig:label:pecan}}
\end{figure}
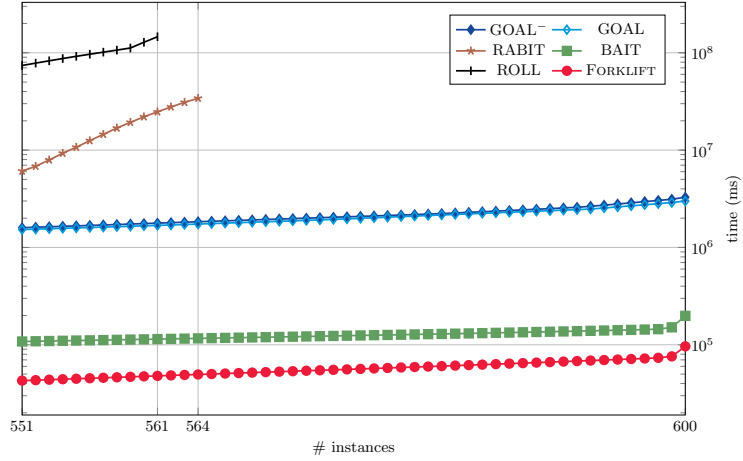
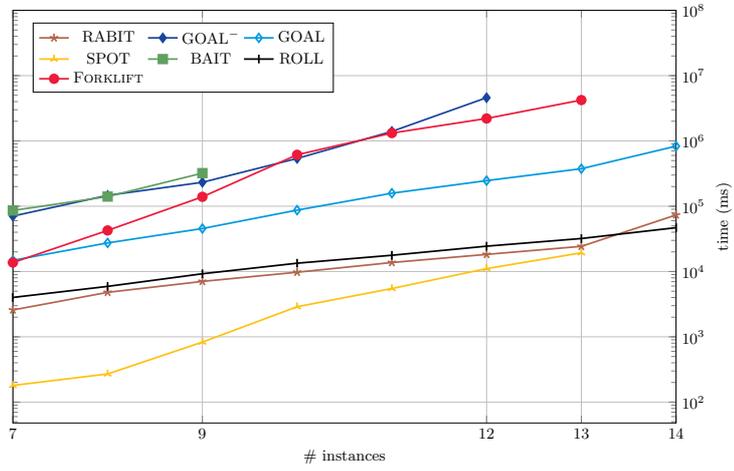
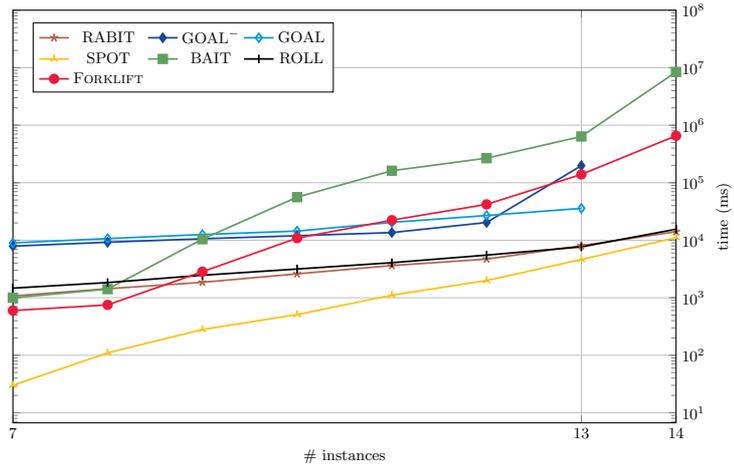
\begin{figure}[ph!]
\centering
\ContinuedFloat
\subfloat[Benchmarks from Ultimate Automizer\label{fig:label:ultatom}]{
	\begin{tikzpicture}[scale=.65]
			\begin{axis}[ylabel=time (ms), ymode=log, yticklabel pos=right, xlabel=\# instances, xmin=551, xmax=600, xtick={551, 561, 564, 600}, tick style={grid=major}, scaled ticks=false, tick label style={/pgf/number format/fixed}, legend style={legend columns=2, legend pos=north east}, width=15cm, height=10cm]
				\addplot[mark size=2.5pt, line width=1pt, mark=diamond*,color=cyan!15!blue] file {ultimate-goal-minus.dat} ;%
				\addplot[mark size=2.5pt, line width=1pt, mark=diamond,color=cyan!85!blue]        file {ultimate-goal.dat} ;%
				\addplot[mark size=2.5pt, line width=1pt, mark=star,color=brown!80!violet]         file {ultimate-rabit.dat} ;%
				\addplot[mark size=2.5pt, line width=1pt, mark=square*,color=green!25!gray] file {ultimate-bainc.dat} ;%
				\addplot[mark size=2.5pt, line width=1pt, mark=|,color=black]     file {ultimate-roll.dat} ;%
				\addplot[mark size=2.5pt, line width=1pt, mark=*,color=magenta!20!red] file {forklift-last-UA.dat}; %
				\legend{GOAL\(^{-}\), GOAL, RABIT, BAIT, ROLL, \FORKLIFT}
			\end{axis}
    \end{tikzpicture}
}\\
\subfloat[Benchmarks from RABIT\label{fig:label:rabit}]{
		\begin{tikzpicture}[scale=.65]
			\begin{axis}[width=15cm, height=10cm, ymode=log, scaled ticks=false, xlabel=\# instances, tick label style={/pgf/number format/fixed}, ymax=10e7, xtick={7,9,12,13,14}, tick style={grid=major}, xmin=7, xmax=14, yticklabel pos=right, legend style={legend columns=3, legend pos=north west}, ylabel=time (ms)]
				\addplot[mark size=2.5pt, line width=1pt, mark=star,color=brown!80!violet] file {rabit-survivor-rabit.dat} ;%
				\addplot[mark size=2.5pt, line width=1pt, mark=diamond*,color=cyan!15!blue] file {rabit-survivor-goal-minus.dat} ;%
				\addplot[mark size=2.5pt, line width=1pt, mark=diamond,color=cyan!85!blue] file {rabit-survivor-goal.dat} ;%
        \addplot[mark size=2.5pt, line width=1pt, mark=Mercedes star, color=yellow!50!orange] file {rabit-spot-2_10_3.dat} ;%
				\addplot[mark size=2.5pt, line width=1pt, mark=square*,color=green!25!gray]  file {rabit-survivor-bait.dat} ;%
				\addplot[mark size=2.5pt, line width=1pt, mark=|,color=black] file {rabit-survivor-roll.dat} ;%
				\addplot[mark size=2.5pt, line width=1pt, mark=*,color=magenta!20!red] file {forklift-last-rabit.dat}; %
				\legend{RABIT, GOAL\(^{-}\), GOAL, SPOT, BAIT, ROLL, \FORKLIFT}
			\end{axis}
		\end{tikzpicture}
}
\\
\subfloat[Benchmarks from RABIT (reduced)\label{fig:label:rabitreduced}]{
		\begin{tikzpicture}[scale=.65]
			\begin{axis}[width=15cm, height=10cm, ymode=log, scaled ticks=false, xlabel=\# instances, tick label style={/pgf/number format/fixed}, ymax=10e7, ylabel=time (ms), xtick={7,13,14}, tick style={grid=major}, xmin=7, xmax=14, yticklabel pos=right, legend style={legend columns=3, legend pos=north west}]
				\addplot[mark size=2.5pt, line width=1pt, mark=star,color=brown!80!violet] file {rabit-survivor-rabit-reduced.dat} ;%
				\addplot[mark size=2.5pt, line width=1pt, mark=diamond*,color=cyan!15!blue] file {rabit-survivor-goal-minus-reduced.dat} ;%
				\addplot[mark size=2.5pt, line width=1pt, mark=diamond,color=cyan!85!blue] file {rabit-survivor-goal-reduced.dat} ;%
        \addplot[mark size=2.5pt, line width=1pt, mark=Mercedes star, color=yellow!50!orange] file {rabit-spot-2_10_3-reduced.dat} ;%
				\addplot[mark size=2.5pt, line width=1pt, mark=square*,color=green!25!gray]  file {rabit-survivor-bait-reduced.dat} ;%
				\addplot[mark size=2.5pt, line width=1pt, mark=|,color=black] file {rabit-survivor-roll-reduced.dat} ;%
				\addplot[mark size=2.5pt, line width=1pt, mark=*,color=magenta!20!red]  file {rabit-survivor-forklift-reduced.dat} ;%
				\legend{RABIT, GOAL\(^{-}\), GOAL, SPOT, BAIT, ROLL, \FORKLIFT}
			\end{axis}
		\end{tikzpicture}
}
\caption{
  Survival plot with a logarithmic \(y\) axis and linear \(x\) axis. Each benchmark has a timeout value of 12h. Parts of the plots left out for clarity. A point is plotted for abscissa value \(x\) and tool \(r\) if{}f \(r\) returns with an answer for \(x\) benchmarks.
  All the failures of BAIT and the one of \FORKLIFT are memory out.%
  \label{fig:experiments}
}
\end{figure}

\paragraph{Analysis.}
It is clear from \figurename~\ref{fig:label:pecan} and \ref{fig:label:ultatom} that \FORKLIFT scales up best on both the Pecan and UA benchmarks.
\FORKLIFT's scalability is particularly evident on the PECAN benchmarks of \figurename~\ref{fig:label:pecan} where its curve is the flattest and no other tool finishes on all benchmarks. 
Note that, in \figurename~\ref{fig:label:ultatom}, the plot for SPOT is missing because we did not succeed into translating the UA benchmarks in the input format of SPOT.
On the UA benchmarks, \FORKLIFT, BAIT and GOAL scale up well and we expect SPOT to scale up at least equally well.
On the other hand, RABIT and ROLL scaled up poorly on these benchmarks.

On the RABIT benchmarks at~\figurename~\ref{fig:label:rabit} both \FORKLIFT and SPOT terminate 13 out of 14 times; BAIT terminates 9 out of 14 times; and GOAL, ROLL and RABIT terminate all the times.
We claim that the RABIT benchmarks can all be solved efficiently by leveraging simulation relations which \FORKLIFT does not use let alone compute.
Next, we justify this claim.
First observe at \figurename~\ref{fig:label:rabit} how GOAL is doing noticeably better than GOAL\(^{-}\) while we have the opposite situation for the Pecan benchmarks \figurename~\ref{fig:label:pecan} and no noticeable difference for the UA benchmarks \figurename~\ref{fig:label:ultatom}.
Furthermore observe how ROLL and RABIT, which both leverage simulation relations in one way or another, scale up well on the RABIT benchmarks but scale up poorly on the PECAN and UA benchmarks.

The reduced RABIT benchmarks at~\figurename~\ref{fig:label:rabitreduced} are obtained by pre-processing every \BA of every RABIT benchmark with the simulation-based reduction operation of SPOT given by \texttt{autfilt --high --ba}.
This preprocessing reduces the state space of the \BAs by more than 90\% in some cases.
The reduction significantly improves how \FORKLIFT scales up (it now terminates on all benchmarks) while it has less impact on RABIT, ROLL and SPOT which, as we said above, already leverage simulation relation internally.
It is also worth noting that GOAL has a regression (from 14/14 before the reduction to 13/14).

Overall \FORKLIFT, even though it is a prototype implementation, is the tool that returns most often (673/674).
Its unique failure %
disappears after a preprocessing using simulation relations of the two \BAs. 
The \FORKLIFT curve for the Pecan benchmarks shows \FORKLIFT scales up best.

Our conclusion from the empirical evaluation is that, in practice \FORKLIFT is competitive compared to the state-of-the-art in terms of scalability.
Moreover the behavior of the \FORQ-based algorithm in practice is far from its worst case exponential runtime.

\section{Discussions}\label{sec:conclusion}

This section provides information that we consider of interest although not essential for the correctness of our algorithm or its evaluation.

\paragraph{Origin of \FORQs.}
Our definition of \FORQ and their suitability property (in particular the language preservation)  are directly inspired from the definitions related to families of right congruences introduced by Maler and Staiger in 1993~\cite{DBLP:conf/stacs/MalerS93} (revised in 2008~\cite{malerSyntacticCongruencesLanguages2008a}).
We now explain how our definition of \FORQs generalizes and relaxes previous definitions~\cite[Definitions~5 and~6]{malerSyntacticCongruencesLanguages2008a}.

First we explain why the \FORQ constraint does not appear in the setting of families of right congruences.
In the context of congruences, relations are symmetric and thus, the \FORQ constraint reduces to \(u \leqI u' \Rightarrow \mathord{\leqF_{u'}} = \mathord{\leqF_u}\).
Therefore the \FORQ constraint trivially holds if the set \(\{\leqF_u\}_{u\in\Sigma^*}\) is quotiented by the congruence relation \(\leqI\), which is the case in the definition~\cite[Definition~5]{DBLP:journals/tcs/MalerS97}.

Second, we point that the condition \(v \leqF_u v' \Rightarrow uv \leqI uv'\) which appears in the definition for right families of congruences~\cite[Definition~5]{malerSyntacticCongruencesLanguages2008a} is not needed in our setting.
Nevertheless, this condition enables an improvement of the \FORQ-based algorithm that we describe next.

\paragraph{Less membership queries.}
We put forward a property of structural \FORQs allowing us to reduce the number of membership queries performed by \FORKLIFT.
Hereafter, we refer to the \emph{picky constraint} as the property of a \FORQ stating \(v \leqF_u v' \Rightarrow uv \leqI uv'\) where \(u, v, v' \in \Sigma^*\). 
We first show how thanks to the picky constraint we can reduce the number of candidate counterexamples in the \FORQ-based algorithm and then, we show that every structural \FORQ satisfies the picky constraint.

In the algorithm of \figurename~\ref{algo:forq}, periods are taken in a basis for the ordering \(\leqF_w\) where \(w \in \Sigma^*\) belongs to a finite basis for the ordering \(\leqI^{-1}\).
The only restriction on \(w\) is that of being comparable to the stem \(u\), as ensured by the test at line \ref{algo:forq:line:U}.
The following lemma formalizes the fact that we could consider a stronger restriction.

\begin{lemma}\label{lemma:goodmax}
	Let \(\mathord{\leqI}\) be a quasiorder  over \(\Sigma^*\) such that \(\mathord{\leqI^{-1}}\) is a right-monotonic well-quasiorder.
  Let \(S, S' \subseteq \Sigma^*\) be such that \(\base[\mathord{\leqI^{-1}}]{S'}{S}\) and \(S'\) contains no two distinct comparable words.
  For all \(u\in\Sigma^*\) and \(v\in\Sigma^+\) such that \(u\in S\) and \(\{wv \st w\in S\} \subseteq S\), there exists \(\mathring{w}\in S'\) such that \(uv^i \leqI \mathring{w}\) and \(\mathring{w}v^j \leqI \mathring{w}\) for some \(i, j \in \N\setminus\{0\}\).
\end{lemma}

As in Section~\ref{sec:crux}, we show that the equivalence \eqref{eq:crux:equivalence} holds but this time for an alternative definition of \(T_\calA\) we provide next.
Given a \(M\)-suitable \FORQ \(\calF \defeq \tuple{\mathord{\leqI}, \{\mathord{\leqF_{u}}\}_{u \in \Sigma^{*}}}\), let
\[
\hat{T}_\calA \defeq \{ \lasso{u}{v} \mid \exists s\in F_\calA\colon u\in U_s,  v\in V_s^w \text{ for some \(w \in W_s\) with \(u\leqI w, wv \leqI w\)}\}
\]
where  for all \(p\in P\) the sets \(U_p\), \(W_p\) and \(\{V_p^w\}_{w\in \Sigma^*}\) such that \(\base[\mathord{\leqI}]{U_p}{\stem_p}\), \(\base[\mathord{\leqI^{-1}}]{W_p}{\stem_p}\) and \(\base[\mathord{\leqF_{w}}]{V_p^w}{\per_p}\) for all \(w \in \Sigma^*\).
Since \(\hat{T}_\calA \subseteq T_\calA \) by definition, it suffices to prove the implication \(\hat{T}_\calA  \subseteq M \Rightarrow \UltimA \subseteq M\).
Let \(\lasso{u}{v} \in \UltimA\), i.e., such that there exists \(s\in F_\calA\) for which \(u \in \stem_s\) and \(v \in \per_s\), satisfying \(uv \leqI u\).
In the context of Lemma~\ref{lemma:goodmax}, taking \(S \defeq \stem_s\) and \(S'\defeq W_s\) fulfills the requirements \(u \in S\) and \(\{wv \st w \in S\} \subseteq S\).
We can thus apply the lemma and ensure the existence of some \(w_0 \in W_s\) satisfying \(uv^i \leqI w_0\) and \(w_0v^j \leqI w_0\) for some \(i, j \in \N\setminus\{0\}\).
Since \(uv^i \in \stem_s\) and \(v^j \in \per_s\) we find that there exist \(u_0\in U_s\) and \(v_0\in V_s^{w_0}\) such that \(u_0 \leqI uv^i\) and \(v_0 \leqF_{w_0} v^j\) thanks to the finite basis property.
We conclude from above that \(v_0 \leqF_{w_0} v^j\), hence that  \(w_0v_0 \leqI w_0v^j\) by the picky condition, and finally that \(w_0v_0\leqI w_0\) by Lemma~\ref{lemma:goodmax} and transitivity.
By definition \(\lasso{u_0}{v_0} \in \hat{T}_\calA\) and the proof continues as the one in Section~\ref{sec:crux} for \(T_\calA\).

To summarize, if the considered \FORQ fulfills the picky constraint then the algorithm of \figurename~\ref{algo:forq} remains correct when discarding the periods \(v\) at line~\ref{algo:forq:line:V} such that \(w v\not\leqI w\).
Observe that discarding one period \(v\) possibly means skipping several membership queries (\(\lasso{u_1}{v}, \lasso{u_2}{v},\ldots\)).
As proved below, the picky constraint holds for all structural \FORQs.

\begin{lemma}\label{lemma:picky}
	Let \(\calB \defeq (Q, q_I, \Delta_\calB, F_\calB)\) be a \BA and \(\calF \defeq \tuple{\mathord{\leqI^\calB}, \{\mathord{\leqF^\calB_{u}}\}_{u\in \Sigma^*}}\) its structural \FORQ.
	For all \(u \in \Sigma^*\) and all \(v, v' \in \Sigma^+\) if \(v \leqF^\calB_u v'\) then \(uv \leqI^\calB uv'\).
\end{lemma}
\begin{proof}
	For all \(q ' \in \runI{uv}\), there exists \(q \in Q\) such that \(\calB \colon q_I \run{u} q \run{v} q'\).
	Hence \((q, q', \zero) \in \runF[{\runI{u}}]{v}\).
	In fact \((q, q', \zero) \in \runF[{\runI{u}}]{v'}\) holds as well since \(v \leqF^\calB_{u} v'\).
	We deduce from the definition of \textRunF that \(\calB \colon q_I \run{u} q \run{v'} q'\) which implies \(q' \in \runI{uv'}\).
	Thus \(\runI{uv} \subseteq  \runI{uv'}\), i.e., \(uv \leqI^\calB uv'\).\qed
\end{proof}

We emphasize that this reduction of the number of membership queries was not included in our experimental evaluation since (1) the proof of correctness is simpler and (2) \FORKLIFT already scales up well without this optimization.
We leave for future work the precise effect of such optimization.%

\paragraph{Why a basis for \(\leqI^{-1}\) is computed?}
Taking periods in a basis for the ordering \(\leqF_w\) where \(w \in \Sigma^*\) is picked in a basis for the ordering \(\leqI^{-1}\) may seem unnatural.
In fact, the language preservation property of \FORQs even suggests that an algorithm without computing a basis for \(\leqI^{-1}\) may exist.
Here, we show that taking periods in a basis for the ordering \(\leqF_u\) where \(u \in \Sigma^*\) is picked in a basis for the ordering \(\leqI\) is not correct.
More precisely, redefining \(T_\calA\) as
\[ \tilde{T}_\calA \defeq \{ \lasso{u}{v} \st \exists s\in F_\calA\colon u\in U_s, v\in V_s^u\}\]
where for all \(p\in P\) we have that \(\base[\mathord{\leqI}]{U_p}{\stem_p}\) and
\(\base[\mathord{\leqF_{w}}]{V_p^w}{\per_p}\) for all \(w \in \Sigma^*\), leads to an \emph{incorrect} algorithm because the equivalence \eqref{eq:crux:equivalence} given by \(\tilde{T}_\calA \subseteq M \iff L(\calA) \subseteq M\) no longer holds as shown below in Example~\ref{ex:break}.

\begin{example} \label{ex:break}
Consider the \BAs given by \figurename~\ref{fig:example:automata}.
We have that \(L(\calA) \nsubseteq L(\calB)\) and, in Example~\ref{ex:running}, we have argued that \(T_\calA = \{\lasso{\varepsilon}{(b)}, \lasso{a}{(a)} \}\) contains the ultimately periodic \(a^\omega\) which is a counterexample to inclusion.
Recall from Example~\ref{ex:running} and \ref{example:struct} that we can set \(U_{p_I}=\{\varepsilon,a\}\) since \(\base[\leqI]{\{\varepsilon,a\}}{\Sigma^*}\), and \(V_{p_I}^{a}=V_{p_I}^{\varepsilon}=\{b\}\) since \(\base[\leqF_{a}]{\{b\}}{\Sigma^+}\) and \(\base[\leqF_\varepsilon]{\{b\}}{\Sigma^+}\).
We conclude from the above definition that \(\tilde{T}_\calA = \{ \lasso{\varepsilon}{(b)}, \lasso{a}{(b)} \}\), hence that \(\tilde{T}_\calA\subseteq L(\calB)\) which contradicts \eqref{eq:crux:equivalence} since \(L(\calA)\nsubseteq L(\calB)\).
\end{example}

\section{Conclusion and future work}

We presented a novel approach to tackle in practice the language inclusion problem between B\"uchi automata.
Our antichain heuristics is driven by the notion of \FORQs that extends the notion of family of right congruences introduced in the nineties by Maler and Staiger~\cite{DBLP:journals/tcs/MalerS97}.
We expect the notion of \FORQs to have impact beyond the inclusion problem, e.g. in  learning~\cite{lmcs:4283} and complementation~\cite{liCongruenceRelationsBuchi2021}.
A significant difference of our inclusion algorithm compared to other algorithms which rely on antichain heuristics, is the increased number of fixpoint computations that, counterintuitively, yield better scalability.
Indeed our prototype \FORKLIFT, which implements the \FORQ-based algorithm, scales up well on benchmarks taken from real applications in verification and theorem proving. 

In the future we want to increase further the search pruning capabilities of \FORQs by enhancing them with simulation relations.
We also plan to study whether \FORQs can be extended to other settings like \(\omega\)-visibly pushdown languages.

\bibliographystyle{splncs04}
\newpage
\appendix
%

\renewenvironment{proof}[1][]{\subsubsection{\ifx&#1&\textbf{Proof.} \else\textbf{#1.} \fi}}{}

\section*{Appendix}

\section{Missing Proofs}

\begin{proof}[Proof of Lemma~\ref{lemma:rcat}]
	We prove the following more general property on the \(n\)th iteration of the function \textRcatA applied on any given vector \(\vec{X} \in \super{\Sigma^*}^{|P|}\):
	\begin{equation*}
		\forall p \in P \quad \catA[n]{\vec{X}}\at{p} = \{ w_1w_2 \in \Sigma^* \st w_1 \in \vec{X}\at{p'}, \calA \colon p' \run{w_2} p, |w_2| \leq n\}\enspace .
	\end{equation*}
	Trivially \(\catA[0]{\vec{X}} = \vec{X}\) for all \(\vec{X} \in \super{\Sigma^*}^{|P|}\).
	Assuming by induction hypothesis that the statement holds for all \(n \in \N\).
	Since \(\catA[n+1]{\vec{X}}\at{p} = \catA{\catA[n]{\vec{X}}}\at{p}\) for all \(\vec{X} \in \super{\Sigma^*}^{|P|}\) and all \(p\in P\), we prove that:
	\begin{align*}
		\catA[n+1]{\vec{X}}\at{p}
		&= \catA[n]{\vec{X}}\at{p} \cup \{ wa  \st w \in \catA[n]{\vec{X}}\at{p'}, a \in \Sigma, \calA \colon p' \run{a} p\}\\
		&= \catA[n]{\vec{X}}\at{p} \cup \bigg\{ w_1w_2a ~\Big|~ \begin{array}{l}
			a \in \Sigma, |w_2| \leq n\\
			w_1 \in \vec{X}\at{p''}, \calA \colon p'' \run{w_2} p' \run{a} p
		\end{array} \bigg\}\\
		&= \catA[n]{\vec{X}}\at{p} \cup \bigg\{ w_1w_2 ~\Big|~ \begin{array}{l}
			1 \leq |w_2| \leq n+1\\
			w_1 \in \vec{X}\at{p'}, \calA \colon p' \run{w_2} p
		\end{array} \bigg\}\\
		&= \{ w_1w_2 \st w_1 \in \vec{X}\at{p'}, \calA \colon p' \run{w_2} p, |w_2| \leq n+1 \}
	\end{align*}

    \noindent
    From the initializations of the vectors \(\vec{U}_0\) and \(\vec{V}_1^{s}\) (\(s\in F_\calA\)) in \figurename~\ref{algo:forq} we find that:
		\[
			\forall p\in P,\vec{U}_0\at{p} = \{ u\in\Sigma^{\leq 0} \st \calA \colon p_I \run{u} p\}\text{ , and }\quad %
			\vec{V}_1^{s}\at{p} = \{ v \in \Sigma \st \calA \colon s \run{v} p\}\enspace .
		\]
    In particular, it is routine to check that \(\vec{U}_0\at{p} = \stem_p \cap \Sigma^{\leq 0}\) and
    \(\vec{V}_1^{s}\at{s} = \per_s \cap \Sigma\) for all \(p\in P,s\in F_\calA\).
	Next we conclude from the above result that for all \(p\in P,s\in F_\calA\):
	\begin{align*}
		\catA[n+1]{\vec{U}_0}\at{p} &= \{ u \in \Sigma^* \st \calA \colon p_I \run{u} p, |u| \leq n+1\} = \stem_p \cap \Sigma^{\leq n+1}\\
		\catA[n+1]{\vec{V}_1^{s}}\at{s} &= \{ v \in \Sigma^+ \st \calA \colon p \run{v} p, |v| \leq n+2\} = \per_s \cap \Sigma^{\leq n+2}\enspace .
	\end{align*}
\qed\end{proof}

\begin{proof}[Proof of Lemma~\ref{lemma:goodmax}]
	Since \(\leqI^{-1}\) is a well-quasiorder, \(S\) admits a finite number of minimums with respect to \(\leqI^{-1}\).
	Also \(S \neq \varnothing\) by hypothesis, implying the non-emptiness of its set of minimums.
	Let \(\mathring{w}_1, \dots, \mathring{w}_k \in S\) be the minimal elements of \(S\) with respect to the well-quasiorder \(\leqI^{-1}\).
	For all \(i \in \{1, \dots, k\}\), we have \(\mathring{w}_i v \in S\) by hypothesis, and thus there exists \(j \in \{1, \dots, k\}\) such that \(\mathring{w}_j \leqI^{-1} \mathring{w}_iv\).
	Hence, the function \(f \colon \{1, \dots, k\} \mapsto  \{1, \dots, k\}\) defined by \(f \colon i \mapsto \min \{j \st \mathring{w}_j \leqI^{-1} \mathring{w}_iv\}\) is well defined.
	In particular, \(\mathring{w}_{f(i)} \leqI^{-1} \mathring{w}_iv\) for all \(i \in \{1, \dots, k\}\) by Skolemization.
	
	Knowing \(u \in S\), we exhibit some \(\ell \in \{1, \dots, k\}\) satisfying \(\mathring{w}_\ell \leqI^{-1} u\).
	Consider the infinite sequence \(\{f^n(\ell)\}_{n\in\N}\) where \(f^0(\ell) \defeq \ell\) and \(f^{n+1}(\ell) \defeq f(f^n(\ell))\).
	We prove by induction that \(\mathring{w}_{f^{i+j}(\ell)} \leqI^{-1} \mathring{w}_{f^i(\ell)}v^j\) for all \(i,j \in \N\).
	Assume that \(j =0\), then \(\mathring{w}_{f^{i}(\ell)} \leqI^{-1} \mathring{w}_{f^i(\ell)}\) holds for all \(i\in\N\) by reflexivity.
	Assume that \(j > 0\) and \(\mathring{w}_{f^{i+j}(\ell)} \leqI^{-1} \mathring{w}_{f^i(\ell)}v^j\).
	The right-monotonicity of \(\leqI^{-1}\) implies \(\mathring{w}_{f^{i+j}(\ell)} v \leqI^{-1} \mathring{w}_{f^i(\ell)}v^{j+1}\).
	By construction \(\mathring{w}_{f^{i+j+1}(\ell)} \leqI^{-1} \mathring{w}_{f^{i+j}(\ell)}v\).
	Hence \(\mathring{w}_{f^{i+j+1}(\ell)} \leqI^{-1} \mathring{w}_{f^{i}(\ell)}v^{j+1}\) holds.
	
	As a consequence of the finiteness of the set \(\{f^n(\ell) \st n \in \N_{\neq 0} \}\), there exists \(x, y \in \N_{\neq 0}\) such that \(0 < x < y\) and \(f^x(\ell) = f^y(\ell)\).
	Thanks to the property previously proved, the following holds.
	\begin{itemize}
		\item  \(\mathring{w}_{f^{y}(\ell)} \leqI^{-1} \mathring{w}_{f^x(\ell)}v^{y-x}\) by taking \(i \defeq x\) and \(j \defeq y-x\).
		\item  \(\mathring{w}_{f^{x}(\ell)} \leqI^{-1} \mathring{w}_{\ell}v^{x}\) by taking \(i \defeq 0\) and \(j \defeq x\).
	\end{itemize}
	Let \(\mathring{w} \defeq \mathring{w}_{f^x(\ell)} = \mathring{w}_{f^y(\ell)}\).
	We have that \(\mathring{w} \leqI^{-1} \mathring{w}v^{y-x}\) and \(\mathring{w} \leqI^{-1} \mathring{w}_\ell v^x\).
	In addition \(\mathring{w}_\ell \leqI^{-1} u\) implies \(\mathring{w} \leqI^{-1} u v^x\) by right-monotonicity and transitivity of \(\leqI^{-1}\).
\qed
\end{proof}

\section{Reductions using Simulation Relations}

See \tablename~\ref{tab:rabitbenchmarks}.
\begin{table}[h!]
	\centering
	\caption{Reduction obtained when preprocessing the B\"uchi automata of the RABIT benchmarks using the \texttt{autfilt} command line utility of SPOT (2.9.7) with the options \texttt{--high --ba}. The numbers between parenthesis indicate the number non deterministic states. SCC stands for strongly connected components.\label{tab:rabitbenchmarks}}
	\begin{tabular}{|l||r@{}r|r|r||r@{}r|r|r|}
		\hline
		B\"uchi aut. & \multicolumn{4}{|c||}{stats before reduction} & \multicolumn{4}{c|}{stats after reduction}\\
		\hline
		& \#states &(nodets) & \#edges & \#SCCs & \#states &(nodets) & \#edges & \#SCCs \\
		\hline\hline
		mcsA& 1\,408 &(1\,240)& 3\,222& 1\,227& 11 &(4)& 15& 1\\
		mcsB& 7\,963 &(7\,819)& 21\,503& 19& 69 &(34)& 108& 1\\
		bakeryA& 1\,510 &(1\,195)& 2\,703& 50& 480 &(235)& 717& 4\\
		bakeryB& 1\,509 &(1\,195)& 2\,702& 49& 481 &(235)& 718& 5\\
		bakeryV2A& 1\,149 &(943)& 2\,090& 49& 647 &(285)& 934& 5\\
		bakeryV2B& 1\,150 &(943)& 2\,091& 50& 702 &(346)& 1\,051& 4\\
		bakeryV3B& 1\,495 &(1\,195)& 2\,675& 44& 518 &(227)& 748& 2\\
		bakeryV3A& 1\,149 &(943)& 2\,090& 49& 647 &(285)& 934& 5\\
		fischerB& 1\,532 &(1\,363)& 3\,850& 1& 175 &(90)& 300& 1\\
		fischerA& 634 &(448)& 1\,395& 1& 8 &(2)& 10& 1\\
		fischerV2B& 56 &(52)& 147& 3& 22 &(8)& 31& 8\\
		fischerV2A& 56 &(52)& 147& 3& 22 &(8)& 31& 8\\
		fischerV3A& 637 &(448)& 1\,400& 4& 10 &(3)& 14& 3\\
		fischerV3B& 638 &(448)& 1\,401& 5& 10 &(2)& 14& 3\\
		fischerV4A& 56 &(52)& 147& 3& 22 &(8)& 31& 8\\
		fischerV4B& 526 &(437)& 1\,506& 4& 268 &(136)& 480& 4\\
		fischerV5B& 643 &(454)& 1\,420& 10& 8 &(2)& 10& 1\\
		fischerV5A& 1\,532 &(1\,363)& 3\,850& 1& 175 &(90)& 300& 1\\
		philsB& 161 &(156)& 482& 1& 108 &(102)& 288& 1\\
		philsA& 23 &(18)& 49& 2& 22 &(12)& 38& 2\\
		philsV2B& 80 &(75)& 212& 2& 21 &(11)& 36& 2\\
		philsV2A& 161 &(156)& 482& 1& 108 &(102)& 288& 1\\
		philsV3A& 161 &(155)& 464& 1& 134 &(126)& 362& 1\\
		philsV3B& 80 &(75)& 212& 2& 21 &(11)& 36& 2\\
		philsV4A& 161 &(156)& 482& 1& 108 &(102)& 288& 1\\
		philsV4B& 161 &(155)& 464& 1& 134 &(126)& 362& 1\\
		petersonB& 20 &(14)& 34& 1& 16 &(9)& 26& 1\\
		petersonA& 20 &(14)& 33& 3& 14 &(6)& 21& 1\\
		\hline
	\end{tabular}
\end{table}
 %

\end{document}